\documentclass[a4paper]{article} 
\usepackage{complexity} 
\usepackage{xspace} 
\usepackage{times}  
\usepackage{helvet}  
\usepackage{courier}  
\usepackage[hyphens]{url}  
\usepackage{graphicx} 
\usepackage{caption} 
\usepackage{algorithm}
\usepackage{algorithmic}
\usepackage{todonotes,enumitem}

\usepackage{cite}

\usepackage{amsfonts} 
\usepackage{mathtools,amsthm,amssymb,amsmath}
\usepackage{tikz}
\usetikzlibrary{shapes}
\usetikzlibrary{arrows,automata, positioning,calc, patterns}
\usepackage{pgfplots}
\usepgfplotslibrary{fillbetween}

\newcommand{\bigoh}{\ensuremath{{\mathcal O}}}
\newcommand{\N}{\ensuremath{{\mathbb N}}}

\usepackage{newfloat}
\usepackage{listings}
\DeclareCaptionStyle{ruled}{labelfont=normalfont,labelsep=colon,strut=off} 
\floatstyle{ruled}
\newfloat{listing}{tb}{lst}{}
\floatname{listing}{Listing}
\newtheorem{theorem}{Theorem} 
\newtheorem{lemma}{Lemma}
\newtheorem{definition}{Definition}

\newcommand{\problem}{{\sc SOAC}\xspace} 
\newcommand{\mproblem}{{\sc mSOAC}\xspace} 

\newcommand{\cmax}{\mathtt{c_{max}}}
\newcommand{\fen}{\ensuremath{\mathtt{fen}}}
\newcommand{\twd}{\ensuremath{\mathtt{twd}}}
\newcommand{\stcw}{\ensuremath{\mathtt{stcw}}}

\makeatletter
\DeclareRobustCommand{\rvdots}{%
  \vbox{
    \baselineskip4\p@\lineskiplimit\z@
    \kern-\p@
    \hbox{.}\hbox{.}\hbox{.}
  }}
\makeatother

\newcommand{\defproblem}[3]{
  \vspace{1mm}
\noindent\fbox{
  \begin{minipage}{0.96\textwidth}
  \begin{tabular*}{\textwidth}{@{\extracolsep{\fill}}lr} #1 \\ \end{tabular*}
  {\bf{Input:}} #2  \\
  {\bf{Question:}} #3
  \end{minipage}
  }
  \vspace{1mm}
}

\usepackage{authblk}

\newcommand{\cut}{{\mathbf{cut}}}
\newcommand{\adh}{{\mathbf{adh}}}
\newcommand{\tor}{{\mathbf{tor}}}
\newcommand{\tcw}{{\mathbf{tcw}}}
\newcommand{\ecw}{{\mathbf{ecw}}}
\newcommand{\loc}{\operatorname{loc}}

\newcommand{\forgottenG}{\mathcal{G}}
\newcommand{\Rec}{\text{Record}}

\renewcommand{\R}{\mathcal{R}}
\renewcommand{\D}{\mathcal{D}}
\newcommand{\Sout}{\mathcal{S}_\emph{out}}
\newcommand{\Sin}{\mathcal{S}_\emph{in}}
\newcommand{\Aout}{\mathcal{A}_\emph{out}}
\newcommand{\Ain}{\mathcal{A}_\emph{in}}
\newcommand{\Imp}{\mathcal{Z}}
\newcommand{\snap}{\text{Snap}}

\title{The Complexity of Optimizing Atomic Congestion\thanks{An extended abstract of this paper appeared in the proceedings of AAAI 2024~\cite{DBLP:conf/aaai/BrandGKM24}.}}

\author[a]{Cornelius Brand}
\author[b]{Robert Ganian}
\author[c]{Subrahmanyam Kalyanasundaram}
\author[b]{Fionn {Mc~Inerney}$^\dagger$}

\affil[a]{Algorithms \& Complexity Theory Group, Regensburg University, Germany}
\affil[b]{Algorithms and Complexity Group, TU Wien, Austria}
\affil[c]{Department of Computer Science and Engineering, IIT Hyderabad, India}

\date{}

\begin{document}

\maketitle

\begin{abstract}
Atomic congestion games are a classic topic in network design, routing, and algorithmic game theory, and are capable of modeling congestion and flow optimization tasks in various application areas.
While both the price of anarchy for such games as well as the computational complexity of computing their Nash equilibria are by now well-understood,
the computational complexity of computing a \emph{system-optimal} set of strategies---that is, a centrally planned routing that minimizes the average cost of agents---is severely understudied in the literature. 
We close this gap by identifying the exact boundaries of tractability for the problem through the lens of the parameterized complexity paradigm. 
After showing that the problem remains highly intractable even on extremely simple networks, we obtain a set of results which demonstrate that the structural parameters which control the computational (in)tractability of the problem are not vertex-separator based in nature (such as, e.g., treewidth), but rather based on edge separators. We conclude by extending our analysis towards the (even more challenging) min-max variant of the problem.
\end{abstract}

{\bf Keywords:} Atomic congestion games, System optimum, Routing, Parameterized complexity, Slim treecut width.

\renewcommand*{\thefootnote}{\fnsymbol{footnote}}
\footnotetext[2]{Corresponding author. Email address: fmcinern@gmail.com. Postal address: Technische Universit\"{a}t Wien, Favoritenstrasse 9-11, E192-01, 1040 Vienna, Austria.}
\renewcommand*{\thefootnote}{\arabic{footnote}}

\section{Introduction}
Congestion games are a by-now classic and widely studied model of network resource sharing.
Introduced by Rosenthal~\cite{Rosenthal73}, congestion games and their innumerable variants and extensions have been the focus of a vast body of literature, spanning fields from algorithmic game theory~\cite{CominettiS0M19} over routing~\cite{KunniyurS03} and network design~\cite{AnshelevichDKTWR04}, to diverse contexts within artificial intelligence~\cite{AshlagiMT07,MeirTBK12,MarchesiC019,HarksHKMS22}, both applied and theoretical.

The basic setup of congestion games comprises a network (modeled as a directed graph) and a set of agents that each have an origin and a destination.\footnote{The setting where there is an infinite number of agents and single agents are infinitesimally small is called \emph{non-atomic}; we focus on the classic, \emph{atomic} case, where agents are individual entities.}
The agents need to decide which routes to take in order to reach their destination in a way that minimizes the cost of their route, where the cost can capture, e.g., the amount of time or resources required.
The eponymous \emph{congestion} enters the scene as follows: the cost accrued by a single agent when traversing a link in the network depends on the number of agents using that link, as described by the link's \emph{latency function}. In essence, the latency function captures how the cost of using each link changes depending on the number of agents using it; depending on the context, more agents using a link could lead to each of them paying a greater cost (e.g., when dealing with traffic congestion) or a lower cost (e.g., when dealing with logistical supply chains), up to a maximum capacity for that link.

It is well-known that selfish strategies in congestion games may not lead to optimal outcomes for all agents,
let alone a \emph{system-optimal} outcome\footnote{In the literature, such outcomes are sometimes called the (\emph{social} or \emph{collective}) \emph{optimum}.} (that is, one achieving the minimum average cost)~\cite{SharonARBS18}. 
In fact, the existence of Nash equilibria for these games was the seminal question investigated by Rosenthal~\cite{Rosenthal73}, and is still of interest in economics and game theory today. 
Notably, the \emph{price of anarchy} for these games has by now been determined~\cite{ChristodoulouK05,AwerbuchAE05}. The price of anarchy in this context is defined as the ratio 
\begin{align} \label{eq:poa}
\tag{PoA}
\sup\nolimits_{S } \frac{\mathrm{cost}(S)}{\mathrm{cost}(S_{\mathrm{sys}})},
\end{align}
where the supremum ranges over all Nash equilibria, $\mathrm{cost}(S)$ is the cost of a set of strategies $S$ for the agents, and $S_{\mathrm{sys}}$ is a \emph{system optimum}, minimizing the average cost over all agents.
In addition, the computational complexity of computing Nash equilibria is equally well-studied~\cite{AckermannRV06,FabrikantPT04}; see also the many recent works on the problem~\cite{HarksHKMS22,0040XPRT22}. 

While the price of anarchy defined in~\eqref{eq:poa} as well as the cost of Nash equilibria $\mathrm{cost}(S)$ (i.e.,~the numerator in~\eqref{eq:poa}) have been extensively treated in the literature on congestion games, it may come as a surprise that, to the best of our knowledge, almost nothing is known about the computational complexity of computing the \emph{denominator} $\mathrm{cost}(S_\mathrm{sys})$ of~\eqref{eq:poa}, that is, determining a \emph{system-optimal set of strategies} for the players. Far from just an intellectual curiosity, applications, e.g., in road or internet traffic routing and planning, make this a pressing question, especially given the rapid developments in autonomous driving systems and the widely adopted political strategy of emphasizing public transportation for its lower environmental impact, which is usually centrally planned and routed, as opposed to individual transport~\cite{sharon2017real,sharon2017network,chen2020path,sharon2021alleviating,jalota2023balancing}. 

One possible reason for this gap may lie in the fact that even the most restricted instances of centrally routing a set of agents across a network in a socially optimal manner become hopelessly hard from the perspective of classical computational complexity theory. To illustrate the severity of this phenomenon, several of these basic classes of intractable instances are described in Section~\ref{sec:nph}.
Another explanation for this blind spot in the literature can possibly be found in the fact that, as far as the price of anarchy is concerned, the network structure itself does not appear to play a significant role, whereas this changes drastically when it comes to actually computing a system-optimal set of strategies.

A particular approach that has proven immensely useful for computationally intractable problems lies in employing the rich toolset offered by \emph{parameterized complexity theory}~\cite{DowneyF13,CyganFKLMPPS15} in order to obtain a rigorous, more fine-grained and detailed description of the computational complexity of a problem. The central aim of such an endeavour is to identify the structural properties of the input---captured via numerical \emph{parameters}---which give rise to fixed-parameter algorithms for the problem (see Section~\ref{sec:prelims}).
Some examples where the parameterized complexity toolset has been successfully applied in the context of Artificial Intelligence research include the series of works on Hedonic Games~\cite{BoehmerE20,BoehmerE20b,GanianHKSS22}, Integer Programming~\cite{GanianO18,EibenGKO19,DvorakEGKO21,ChanCKKP22}, Data Completion~\cite{GanianKOS18,DahiyaFPS21,GanianHKOS22,KoanaFN23}, and Bayesian Network Learning~\cite{OrdyniakS13,GruttemeierKM21,GanianK21,GruttemeierK22}. 

\paragraph{Our Contributions.}
As mentioned above, the problem of computing system-optimal strategies in atomic congestion games (\problem) is extremely hard in terms of classical complexity theory. The core approach of parameterized complexity analysis is to identify those structural properties that a problem instance should have that, even though exceedingly hard in the general case, give rise to its fixed-parameter tractability. Since the network is modeled as a graph, a natural first choice would be to parameterize by the well-established \emph{treewidth} (of the underlying undirected graph)~\cite{RobertsonS86}. Unfortunately, as our first result, we show that the problem remains \NP-hard not only on networks which have treewidth $2$, but even on networks consisting of a star plus an additional vertex (Theorem~\ref{hard-vc}). This result rules out not only the use of treewidth, but also of virtually all other reasonable graph measures, as a single parameter to solve the problem in full generality. 

The above lower bound essentially means that, in order to achieve progress, one needs to combine structural parameters with some auxiliary parameterization of the problem instances. In the context of congestion games, it would seem tempting to consider the number of agents as such an auxiliary parameter, however that would severely restrict any obtained algorithms: while one could reasonably expect that networks of interest may be well-structured, the number of agents in relevant instances of congestion games is typically large, and hence, does not constitute a well-applicable parameter. Instead, here we consider the maximum capacity $\cmax$ of a link in the network as an auxiliary parameter---a value which is never larger, but could be much lower, than the total number of agents in the whole network.

It is important to note that \problem\ remains extremely challenging even when parameterized by $\cmax$. In fact, even if $\cmax$ is fixed to a small constant, the problem is \NP-hard when restricted to networks of constant treewidth (Theorem~\ref{hard-td}); the same reduction also rules out the use of other network parameters based on decompositions along \emph{small vertex separators} (such as \emph{treedepth}~\cite{sparsity}).
However, as we show in our main algorithmic result (Theorem~\ref{thm:stcwfpt}), \problem\ is fixed-parameter tractable when parameterized by $\cmax$ plus a suitable measure that guarantees the network's decomposability along \emph{small edge cuts}. 

Basic examples of such measures include the treewidth plus the maximum degree of the network~\cite{OrdyniakS13,GozupekOPSS17}, or the
 \emph{feedback edge number} (i.e., the edge deletion distance to acyclic networks)~\cite{KoanaKNNZ21,FuchsleMNR22}. In our contribution, we build on the recently introduced \emph{spanning tree} decompositions~\cite{GanianK21,GanianK22} to provide a significantly more general result---in particular, we develop
a highly non-trivial dynamic programming algorithm to establish fixed-parameter tractability with respect to the recently introduced slim variant of \emph{treecut width}. We complement Theorem~\ref{thm:stcwfpt} with lower bounds (Theorems~\ref{hard-dags}, \ref{hard-tw-delta}, and ~\ref{W1-tcw}) that show the result to be an essentially tight delimitation of the exact boundaries of tractability for \problem.

In the final section of this article, we focus on a more general variant of \problem, where instead of requiring all agents to be routed to their destinations, we ask to minimize the system optimum while routing as many agents as possible (or, equivalently, when at most $\alpha$ agents may be left unrouted).
This problem, motivated in part by similar lines of investigation conducted for, e.g., multi-agent path finding~\cite{Huang0KD22} and vehicle routing~\cite{PhamHVN22,AbuMonsharA22}, can be seen as a ``min-max'' variant of \problem, and hence, we denote it \mproblem. Crucially, \mproblem
is even more challenging than \problem: it remains \NP-hard on bidirected trees even for $\cmax=1$~\cite{ErlebachJansen01}, and, perhaps even more surprisingly, is left open on very simple network structures, such as bounded-capacity star networks, when parameterized by $\alpha$. Be that as it may, as our final result, we show that on bounded-degree networks, the spanning-tree based algorithm obtained for \problem\ can be lifted to also solve \mproblem\ via a fixed-parameter algorithm when parameterized by the treewidth of the network, along with $\alpha$ and $\cmax$.

A mind-map of our results is provided in Figure~\ref{fig:mindmap}.

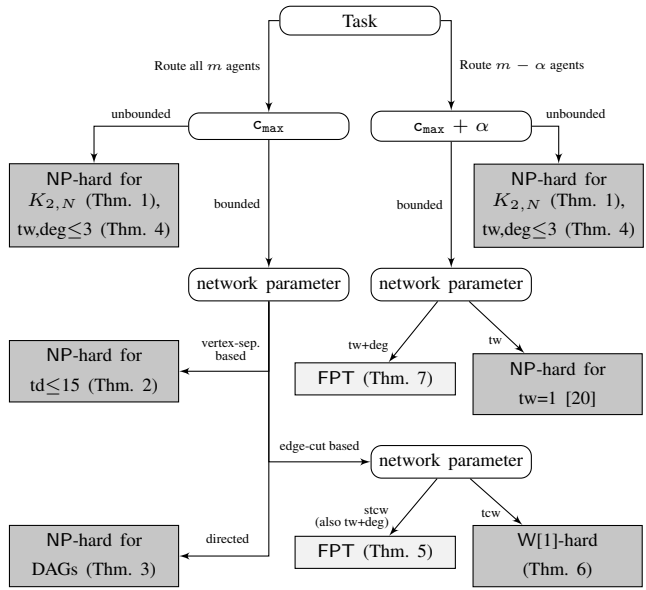
\begin{figure}
\tikzstyle{block} = [rectangle, draw, fill=white!20, 
    text width=6em, text centered, rounded corners, minimum height=1em]
\tikzstyle{line} = [draw, -latex']
\tikzstyle{cloud} = [draw, text centered, rectangle,fill=gray!10, node distance=0cm, text width=6em,
    minimum height=0.0em]
\tikzstyle{cloudhard} = [draw, text centered, rectangle,fill=gray!45, node distance=0cm, text width=6.37em,
    minimum height=0.0em]
    \centering
    \begin{tikzpicture}[node distance =0cm, auto]
    \small
    \centering
    \node [block] (routing) {\scriptsize{Task}};
    \node [block, left of=routing, below of=routing, yshift=-0.2cm, node distance=1.2cm] (cmax1) {\scriptsize{$\cmax$}};
    \node [block, right of=routing, below of=routing, yshift=-0.2cm, node distance=1.2cm] (cmax2) {\scriptsize{$\cmax+\alpha$}};
    \node [block, below of=cmax1, node distance=2.1cm] (sep) {\scriptsize{network parameter}};
    \node [cloudhard, left of=cmax1, below of=cmax1, xshift=-0.1cm, yshift=1.151cm, node distance=2.2cm] (hard1) {\scriptsize \NP-hard for $K_{2,N}$ (Thm. \ref{hard-vc}), \\ tw,deg$\leq$3 (Thm.~\ref{hard-tw-delta})};
    \node [cloudhard, below of=hard1, node distance=2.2cm] (hard2) {\scriptsize \NP-hard for td$\leq$15 (Thm.~\ref{hard-td})};
    \node [cloudhard, below of=hard2, node distance=2.45cm] (harddag) {\scriptsize \NP-hard for DAGs (Thm.~\ref{hard-dags})};
    \node [draw=none, below of=sep, node distance = 1.5cm] (branch) {};
    \node [cloudhard, below right of=cmax2, yshift=0.36cm, node distance=2cm] (hard4) {\scriptsize \NP-hard for $K_{2,N}$ (Thm. \ref{hard-vc}), \\ tw,deg$\leq$3 (Thm.~\ref{hard-tw-delta})};
    \node [block, below of=cmax2, node distance=2.1cm] (param2) {\scriptsize network parameter};
    \node [cloud, left of=param2, below of=param2, node distance=0.85cm, xshift=-0.15cm,yshift=-0.4cm] (fpt2) {\scriptsize \FPT{} (Thm. \ref{thm:msoacfpt})};
    \node [draw=none, right of= cmax2, node distance=2cm] (branch2) {};
    \node [draw=none, below of= param2, node distance=0.85cm] (branch3) {};
    \node [cloudhard, right of = branch3, node distance=1.4cm,yshift=-0.47cm] (hard5) {\scriptsize \NP-hard for tw=1 \cite{ErlebachJansen01}};
    \node [block, right of=branch, below of=param2, yshift=-1.35cm, xshift=-1cm, node distance=1cm] (param) {\scriptsize network parameter};
    \node [cloudhard, right of=param, below of=param, xshift=0.3cm, yshift=-0.15cm, node distance=1.1cm] (hard3) {\scriptsize \W[1]-hard (Thm. \ref{W1-tcw})};
    \node [cloud, left of=param, below of=param, xshift=0.05cm, yshift=-0.15cm, node distance=1.05cm] (fpt1) {\scriptsize \FPT{} (Thm. \ref{thm:stcwfpt})};
    
    \path [line] (routing) -| node [near end, left] {\tiny{Route all $m$ agents}} (cmax1);
    \path [line] (routing) -| node [near end,align=center] {\tiny Route $m-\alpha$ agents} (cmax2);
    \path [line] (cmax1.south) -| node [near end,left] {\tiny bounded} (sep);
    \path [line] (cmax1) -| node [near start, above] {\tiny unbounded} (hard1);
    \path [line] (sep.south) |- node [near start, below, xshift=-0.5cm, yshift=0.3cm, label={[align=center,font=\tiny\linespread{0.8}\selectfont]below:\tiny vertex-sep. \\ \tiny based}] {} (hard2);
    \path [line] (sep.south) |- node [near end, above] {\tiny edge-cut based} (param);
    \path [line] (sep.south) |- node [near end, above, xshift=0.05cm, yshift=0.4cm, label={[align=center]below:\tiny directed}] {} (harddag);
    \path [line] (param) -- node [near end, left] {\tiny tcw} (hard3);
    \path [line] (param) -- node [near end, left, align=right,font=\tiny\linespread{0.8}\selectfont,yshift=0.05cm,xshift=-0.05cm] {\tiny stcw\\ \tiny (also tw+deg)} (fpt1);
    \path [line] (cmax2.east) -| node [near start, above right, xshift=-0.1cm] {\tiny unbounded} (hard4);
    \path [line] (cmax2) -- node [left] {\tiny bounded}  (param2);
    \path [line] (param2) -- node [near end, left,xshift=-0.05cm,yshift=0.02cm] {\tiny tw+deg} (fpt2);
    \path [line] (param2) -- node [near end, left] {\tiny tw} (hard5);
	\end{tikzpicture}
\caption{A mind map of our results on computing system-optimal strategies in congestion games. The formal problem definition as well as a discussion of the considered parameters is provided in Section~\ref{sec:prelims}; here, tw stands for treewidth, deg stands for maximum degree, td stands for \emph{treedepth}, (s)tcw stands for (slim) treecut width, and DAGs stands for directed acyclic graphs.
}
\label{fig:mindmap}
\end{figure}

\section{Preliminaries}\label{sec:prelims}

For a positive integer $i\in \mathbb{N}$, we let $[i]=\{1,2,\dots,i\}$.
We refer to the book by Diestel~\cite{Diestel12} for
standard graph terminology. 
While the networks are modeled as directed graphs (i.e., digraphs), many of our results use the skeletons (i.e., the underlying undirected graphs) of these digraphs. The \emph{skeleton} $\underline G$ of a directed graph $G$ is the simple undirected graph obtained by replacing each arc in $G$ by an undirected edge.
The graph class $K_{i,N}=\{K_{i,j}~|~j\in \mathbb{N}\}$ is the class of all complete bipartite graphs where one side has size $i$.

\paragraph{Formal Problem Definition.}
Given a digraph $G=(V,E)$, let $\mathcal{P}$ be the set of all directed paths in $G$.  Given a set $A=\{a_1,\dots,a_m\}$ of agents where each agent $a_i$ is associated with a tuple $(s_i,t_i)\in V^2$, a \emph{flow assignment} $F$ is a mapping from $A$ to $\mathcal{P}$ such that each agent $a_i$ is mapped to a directed path from $s_i$ to $t_i$. For an arc $e\in E$, let $f_F(e)=|\{a~|~a\in A \wedge e\in F(a)\}|$ be the number of agents whose flow passes through $e$; when $F$ is clear from the context, we omit it and simply use $f(e)$ instead.

Intuitively, our primary problem of interest asks to compute a flow assignment of all the agents that minimizes the total cost. However, in order to formalize the algorithmic lower bounds obtained for the problem, we follow the standard practice of formulating the problem as an equivalent decision problem (see below). To avoid any doubts, we remark that all our algorithmic results are constructive and can also immediately output the minimum cost of a flow assignment with the required properties.

\defproblem{System Optimum Atomic Congestion (\problem)}{A digraph $G=(V,E)$, a positive integer $\lambda$, a set $A=\{a_1,\dots,a_m\}$ of agents where each agent $a_i$ is associated with a tuple $(s_i,t_i)\in V^2$, 
and for each arc $e\in E$, a latency function $\ell_e: [m] \rightarrow \mathbb{R}_{\geq 0} \cup \{\infty\}$.}{Does there exist a flow assignment $F$ such that $\sum_{e\in E}f_F(e)\cdot \ell_e(f_F(e)) \leq \lambda$?}

In specific settings studied in the literature, the latency function is sometimes required to satisfy certain additional conditions, such as being non-decreasing.
As illustrative examples, observe that when agents represent individual vehicles in a traffic network, the latency function will typically be increasing (the cost of 100 agents using a single link is greater than 100 times the cost of that link when it is used by a single agent), but if agents represent individual parcels or shipments in a logistics network, one would expect it to be decreasing (the cost of 100 agents using a single link would be lower than 100 times the cost of that link when it is used by a single agent).
In order to capture as wide a range of scenarios as possible---including, e.g.,~buffered or batch-wise processing at network nodes leading to decreasing or even oscillating latencies, respectively---our study targets the problem with arbitrary latency functions.

Given a latency function $\ell_e: [m] \rightarrow \mathbb{R}_{\geq 0}\cup\{\infty\}$ for an arc $e$, let the capacity $c_e$ of $e$ be defined as the maximum admissible value of the latency function, i.e.,~$\max\{z~|~\ell_e(z)\neq \infty\}$. Let the \emph{maximum capacity} of a network be defined as $\cmax:=\max_{e\in E}c(e)$. Crucially, while the maximum capacity can never exceed the number $m$ of agents in the network, one can reasonably expect it to be much smaller than $m$ in more complex networks.

\paragraph{Parameterized Complexity Theory.}
In parameterized complexity~\cite{CyganFKLMPPS15,DowneyF13,Niedermeier06}, the
running-time of an algorithm is studied with respect to a parameter
$k\in\mathbb{N}$ and input size~$n$. The basic idea is to find a parameter
that describes the structure of the instance such that the
combinatorial explosion can be confined to this parameter. In this
respect, the most favorable complexity class is \FPT\
(\textit{fixed-parameter tractable}), which contains all problems that
can be decided by an algorithm running in time $f(k)\cdot
n^{\bigoh(1)}$, where $f$ is a computable function. Algorithms with
this running-time are called \emph{fixed-parameter algorithms}. 
A basic way of excluding fixed-parameter tractability for a parameterized problem is to show that it remains \NP-hard even when the parameter $k$ is bounded by a constant. However, this is not always possible: some problems can be solved in non-uniform polynomial time for each fixed value of the parameter. In these cases, establishing hardness for the complexity class $\W[1]$ via a \emph{parameterized reduction} rules out the existence of a fixed-parameter algorithm under the well-established assumption that $\W[1]\neq \FPT$. A parameterized reduction can be thought of as a classical polynomial-time reduction, but with the distinction that it (1) can run in fixed-parameter time instead of polynomial time, and (2) must bound the parameter of the output instance by a function of the parameter of the initial instance.

A natural and typical direction in parameterized complexity is to consider parameters tied to the structural properties of the instance. In many cases, it turns out that achieving fixed-parameter tractability is only possible if one parameterizes by multiple properties of the instances simultaneously; formally, this is simply reflected by setting the parameter to be the sum of both values. As a basic example, the classical \textsc{Independent Set} problem on graphs is not fixed-parameter tractable w.r.t.\ the size $\ell$ of the sought set nor the maximum degree $d$ of the graph, but is well-known to admit a fixed-parameter algorithm w.r.t.\ $\ell+d$~\cite{CyganFKLMPPS15}.

\paragraph{Structural Parameters.}
In this work, we identify the exact boundaries of tractability for \problem\ in the context of fundamental graph parameters, as depicted in Figure~\ref{fig:pararel}.
For two of these parameters---notably treewidth~\cite{RobertsonS86} and treedepth~\cite{sparsity}---we do not need to provide explicit definitions since the respective lower bounds we obtain for \problem\ construct instances whose skeletons are well-known to have bounded treewidth and treedepth. The \emph{feedback edge number}~\cite{KoanaKNNZ21,FuchsleMNR22} is simply the minimum number of edges that need to be removed from an undirected graph in order to obtain a forest.
The term ``directed graph parameters'' here broadly refers to all parameters that achieve constant values on directed acyclic graphs; this includes directed treewidth~\cite{JohnsonRST01}, Kelly-width~\cite{hunter2008digraph}, and DAG-width~\cite{BerwangerDHKO12}, to name a few. 

The proofs of our main results for \problem\ (Theorem~\ref{thm:stcwfpt}) and \mproblem\ (Theorem~\ref{thm:msoacfpt}), as well as the lower bound in Theorem~\ref{W1-tcw}, will require a more in-depth introduction of treecut width and its slim variant.
A {\em treecut decomposition} of $G$ is a pair $(T,\mathcal{X})$ which consists of a rooted tree $T$ and a near-partition $\mathcal{X}=\{X_t\subseteq V(G): t\in V(T)\}$ of $V(G)$, where a near-partition is a partitioning of a set which can also contain the empty set. A set in the family $\mathcal{X}$ is called a {\em bag} of the treecut decomposition.

For any node $t$ of $T$ other than the root $r$, let $e(t)=ut$ be the unique edge incident to $t$ on the path to $r$. Let $T^u$ and $T^t$ be the two connected components in $T-e(t)$ which contain $u$ and $t$, respectively. Note that $(\bigcup_{q\in T^u} X_q, \bigcup_{q\in T^t} X_q)$ is a near-partition of $V(G)$, and we use $\cut(t)$ to denote the set of edges with one endpoint in each part. We define the {\em adhesion} of $t$ ($\adh_T(t)$ or $\adh(t)$ in brief) as $|\cut(t)|$; if $t$ is the root, we set $\adh_T(t)=0$ and $\cut(t)=\emptyset$.

The {\em torso} of a treecut decomposition $(T,\mathcal{X})$ at a node $t$, written as $H_t$, is the graph obtained from $G$ as follows. If $T$ consists of a single node $t$, then the torso of $(T,\mathcal{X})$ at $t$ is $G$. Otherwise, let $T_1, \ldots , T_{\ell}$ be the connected components of $T-t$. For each $i=1,\ldots , \ell$, the vertex set $Z_i\subseteq V(G)$ is defined as the set $\bigcup_{b\in V(T_i)}X_b$. The torso $H_t$ at $t$ is obtained from $G$ by {\em consolidating} each vertex set $Z_i$ into a single vertex $z_i$ (this is also called \emph{shrinking} in the literature). Here, the operation of consolidating a vertex set $Z$ into $z$ is to substitute $Z$ by $z$ in $G$, and for each edge $e$ between $Z$ and $v\in V(G)\setminus Z$, adding an edge $zv$ in the new graph. We note that this may create parallel edges.

The operation of {\em suppressing} a vertex $v$ of degree at most $2$ consists of deleting~$v$, and when the degree is two, adding an edge between the neighbors of $v$. Given a connected graph $G$ and  $X\subseteq V(G)$, let the {\em 3-center} of $(G,X)$ be the unique graph obtained from $G$ by exhaustively suppressing vertices in $V(G) \setminus X$ of degree at most two. Finally, for a node $t$ of $T$, we denote by $\tilde{H}_t$ the 3-center of $(H_t,X_t)$, where $H_t$ is the torso of $(T,\mathcal{X})$ at $t$. Let the \emph{torso-size} $\tor(t)$ denote $|\tilde{H}_t|$. 

\begin{definition}
The width of a treecut decomposition $(T,\mathcal{X})$ of $G$ is defined as $\max_{t\in V(T)}\{ \adh(t), \tor(t) \}$. The treecut width of $G$, or $\tcw(G)$ in short, is the minimum width of $(T,\mathcal{X})$ over all treecut decompositions $(T,\mathcal{X})$ of $G$.
\end{definition}

The graph parameter \emph{slim treecut width} can be defined analogously to treecut width, but with the distinction that suppressing only occurs for vertices of degree at most $1$. Yet, for our algorithmic applications, it will be more useful to use a different characterization of the parameter---one based on spanning trees and much better suited to the design of dynamic programming algorithms.

For a graph $G$ and a tree $T$ over $V(G)$, let the \emph{local feedback edge number} at $v\in V(G)$ be $E_{\loc}^{G,T}(v)=\{uw\in E(G)\setminus E(T)~|~\text{the path between $u$ and  $w$ in $T$ contains }v\}.$
The \emph{edge-cut width} of the pair $(G,T)$ is $\ecw(G,T)=1+\max_{v\in V} |E_{\loc}^{G,T}(v)|$.

\begin{lemma}[Prop.\ 27 and Thm.\ 30, \cite{GanianK22}]
\label{lem:computestcw}
Every graph $G$ with slim treecut width $k$ admits a spanning tree $T$ over some supergraph $G'$ of $G$ such that $(G',T)$ has edge-cut width at most $3(k+1)^2$. Moreover, such a pair $(G',T)$ can be computed in time $2^{k^{\bigoh(1)}}\cdot |V(G)|^4$.
\end{lemma}

\begin{lemma}[Prop.\ 22, Prop.\ 26, and Thm.\ 30,~\cite{GanianK22}]
\label{lem:0tcw}
Every graph $G$ with maximum degree $d$ and treewidth $w$ admits a spanning tree $T$ over some supergraph $G'$ of $G$ such that both the edge-cut width of $(G',T)$ and the maximum degree of $T$ are upper-bounded by $\bigoh(d^2w^2)$. Moreover, such a pair $(G',T)$ can be computed in time $2^{(dw)^{\bigoh(1)}}\cdot |V(G)|^4$.
\end{lemma}

We refer to the work that introduced the notion~\cite{GanianK22} for a broader discussion of slim treecut width, as well as illustrative figures.

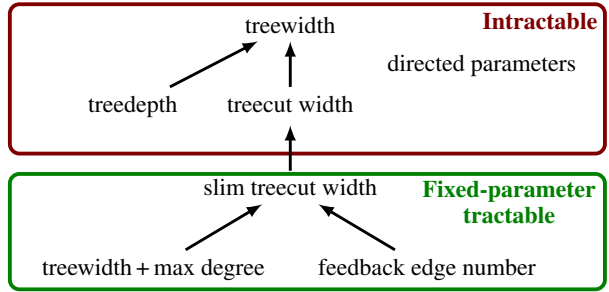
\begin{figure}
\vspace{0.3cm}
\centering
\scalebox{0.9}{
\begin{tikzpicture}[scale = 1]

\draw[rounded corners,red!50!black,ultra thick] (-0.1, -1.9) rectangle (8.6, 0.3) {};
\draw[rounded corners,green!50!black,ultra thick] (-0.1, -3.9) rectangle (8.6, -2.2) {};

\node (TW) at (4, 0) {treewidth};
\node (TD) at (1.7, -1.2)  {treedepth};
\node (TCW) at (4, -1.2)  {\phantom{p}treecut width\phantom{p}};
\node (DP) at (6.8, -0.6)  {directed parameters};

\node (STCW) at (4, -2.4)  {slim treecut width};
\node (TWD) at (2, -3.6)  {treewidth\,+\,max degree};
\node (FEN) at (6, -3.6)  {feedback edge number};

\draw[-latex,very thick] (TD) -- (TW);
\draw[-latex,very thick] (TCW) -- (TW);
\draw[-latex,very thick] (STCW) -- (TCW);
\draw[-latex,very thick] (TWD) -- (STCW);
\draw[-latex,very thick] (FEN) -- (STCW);

\node (fpt) at (7.2, -2.45)  {\color{green!50!black} \bf Fixed-parameter};
\node (fpt) at (7.2, -2.8)  {\color{green!50!black} \bf tractable};
\node (fpt) at (7.65, 0.05)  {\color{red!50!black} \bf Intractable};

 \end{tikzpicture}}
 \caption{A pictorial view of the tractability of \problem when capacities are bounded. 
An arc from a parameter $x$ to a parameter $y$ indicates that $x$ is dominated by $y$, i.e., a bound on $x$ implies a bound on $y$ but the opposite does not hold.}
 \label{fig:pararel}
\end{figure}

\section{The Surprising Difficulty of Solving Atomic Congestion Games}
\label{sec:nph}
Before initializing our more fine-grained parameterized analysis of \problem, we first 
consider the computational complexity of the problem from the classical point of view, that is, with respect to \NP-hardness. In fact, we show that even surprisingly simple input instances turn out to be intractable.

To begin, we show that \problem\ is \NP-hard even when the underlying undirected graph of the input digraph is restricted to the class $K_{2,N}$. We prove this result (as well as Theorem~\ref{hard-tw-delta} later on) via a reduction from the following variant of the classical \textsc{Subset Sum} problem over $d$-dimensional vectors: 

\defproblem{Multidimensional Uniform 0/1 Knapsack (\textsc{MUKS})}
{A set $S = \{\vec{v}_1,\ldots,\vec{v}_n\} \subseteq \N^d$ of $d$-dimensional vectors containing natural numbers, a positive integer $k$, and a target vector $\vec T\in\N^d$.}
{Is there a subset $S' \subseteq S$ of at least $k$ vectors such that $\sum_{\vec s \in S'} \vec s \leq \vec T$?}

\begin{lemma}
\label{lem:muks}
\textsc{MUKS} is \NP-hard and also \W\textup{[1]}-hard parameterized by $d$, even if all numbers are encoded in unary.
\end{lemma}

\begin{proof}
The reduction proceeds from a problem called \textsc{Multidimensional Relaxed Subset Sum} (MRSS),
which is \NP-hard even if the inputs are encoded in unary~\cite{GanianOS17}.
In its definition, it subtly differs from MUKS above. Namely, the direction of the inequalities are swapped, that is, one seeks a subset of \emph{at most} $k$ vectors that sums up to \emph{at least} $\vec{T}$.
Now, given an instance $(S, \vec{T}, k)$ of MRSS, let $\vec{\Sigma} := \sum_{s \in S} s$ be the total sum over all vectors.
The instance of MUKS then consists of $(S, \vec{\Sigma} - \vec{T}, n - k)$.
We claim that the MUKS instance has a solution of size at least $n-k$ if and only if the original MRSS-instance has a solution of size at most $k$.
Indeed, for $S' \subseteq S$, $\sum_{s \in S'} s \geq \vec{T}$ is equivalent to $\vec{\Sigma} - \sum_{s \in S'} s \leq \vec{\Sigma} - \vec{T}$, and, by definition, the left-hand side equals $\sum_{s \in S''} s$, where $S'' = S - S'$ is a set of at least $n-k$ vectors. This proves the correctness of the reduction, and hence, the \NP-hardness on unary inputs.

As for the \W[1]-hardness of the problem, we note that MRSS is also known to be \W[1]-hard when parameterized by $d$~\cite{GanianOS17}. Our polynomial-time reduction does not change the value of $d$, and hence, it immediately also establishes the \W[1]-hardness of MUKS w.r.t.\ $d$.
\end{proof}

Before proceeding to the reduction, let us briefly remark on the significance of the result in the context of this article. Essentially, establishing \NP-hardness on $K_{2,N}$ rules out not only fixed-parameter tractability under almost all commonly considered graph parameterizations (including not only the classical \emph{treewidth}~\cite{RobertsonS86}, but also the \emph{treedepth}~\cite{sparsity}, the \emph{vertex cover number}~\cite{KorhonenP15}, and the more recently introduced edge-cut variants of treewidth like the \emph{treecut width}~\cite{MarxW14}), but even polynomial-time algorithms for instances where such parameters are bounded by a constant.

\begin{theorem}
\label{hard-vc}
\problem\ is \NP-hard even when restricted to networks whose skeletons belong to the class $K_{2,N}$.
\end{theorem}

\begin{proof}
We provide a polynomial-time reduction from $\textsc{MUKS}$ where numbers are encoded in unary to \problem. Recall that $\textsc{MUKS}$ is \NP-hard by Lemma~\ref{lem:muks}. The reduction is as follows.
For every input $\vec v_i \in S$ in the MUKS instance, there is a source vertex $s_i$.
For each of the $d$ entries $T_1,\ldots,T_d$ in the target vector $\vec{T}$, there is a target vertex $t_j$.
There are also two vertices $h_0$ and $h_1$.
Each of the $n$ source vertices $s_i$ is connected by outgoing arcs to both $h_0$ and $h_1$.
The latencies are defined as $\ell_{s_i,h_0} = 0$ for all arcs $(s_i,h_0)$, and as $\ell_{s_i,h_1}(x) = 1/x$ for all arcs $(s_i, h_1)$.
Both $h_0$ and $h_1$ have outgoing arcs to all the target vertices $t_j$. For arcs $(h_0,t_j)$,
$\ell_{h_0,t_j}(x) = 0$ if $x \leq T_j$, and otherwise, $\ell_{h_0,t_j}(x) = \infty$.
That is, the capacity of the arc $(h_0, t_j)$ is equal to $T_j$, for each $j\in [d]$.
For arcs $(h_1,t_j)$, $\ell_{h_1,t_j}(x) = 0$.

For the $i$-th input vector $\vec{v}_i$, let $v_{i,j}$ be its $j$-th entry. 
For all $i\in [n]$ and $j\in [d]$, 
we want to route $v_{i,j}$ agents from $s_i$ to $t_j$.
That is, there are $v_{i,j}$ copies of the pair $(s_i, t_j)$ for all $i\in [n]$ and $j\in [d]$. 

We claim that there is a solution of cost at most $n-k$ if and only if the original instance of \textsc{MUKS} has a solution comprised of at least $k$ vectors. First, any solution $S'\subset S$ for \textsc{MUKS} translates immediately into a solution of cost at most $n-k$:
the choice of latency $\ell_e(x) = 1/x$ ensures that $f(e) \ell_e(f(e)) = x \cdot 1/x = 1$, where $f(e) = x$. 
Hence, the cost of a solution is just $n-g$, where $g$ is the number of \emph{distinct origins} from which \emph{all} agents are routed over $h_0$ (since every origin for which at least one agent is routed via $h_1$ adds $1$ to the cost).
Therefore, by construction, routing all agents from every $s_i$ via $h_0$ whenever $\vec v_i \in S'$ transforms $S'$ into a solution of cost at most $n-k$; in particular, the capacities of the edges from $h_0$ are respected.

On the other hand, suppose we are given a set of strategies of total cost at most $n-k$. 
This means that there are precisely $v_{i,j}$ agents routed from $s_i$ to $t_j$ for all $i\in [n]$ and $j\in [d]$.
Again, by the property of the latency function $1/x$, there is a solution to the instance of \problem of equal cost such that either none or all of the agents from a single origin are routed via $h_1$. For, suppose no agents are routed via $h_1$, then there is nothing to show.
Otherwise, if at least one agent is routed via $h_1$, then the cost added by this is already $1$, and routing all the remaining agents that originate in the same origin via $h_1$ will not add costs (and can only ever relax the capacity constraints). 
Hence, we can assume that all the agents from the same origin travel along the same vertex $h_i$ for $i \in \{0,1\}$ to their respective destinations in the solution. But, this means that there are at least $k$ distinct origins that can be routed via $h_0$ and respect the capacity constraints, and hence, form a valid solution for the \textsc{MUKS} instance.

To complete the proof, note that the skeleton of the digraph in the constructed instance of \problem{} above is $K_{2,d+n}$.
\end{proof}

Since Theorem~\ref{hard-vc} essentially rules out fixed-parameter algorithms based on structural network parameters alone, we turn our attention to parameterizing by the maximum capacity $\cmax$. As our first result in this direction, we show that \problem\ remains \NP-hard on networks of constant treewidth, even if the maximum capacity of a link is just $1$.
This is done via a reduction from the following classical graph problem.

\defproblem{Edge-Disjoint Paths (\textsc{EDP})}
{A graph $G$ and a set $P$ of terminal pairs, that is, a set of subsets of $V(G)$ of size two.}
{Is there a set of pairwise edge-disjoint paths connecting each set of terminal pairs in $P$?}

\begin{theorem}
\label{hard-td}
\problem\ is \NP-hard even when restricted to networks with $\cmax=1$ and whose skeletons have treewidth or treedepth at most $15$.
\end{theorem} 

\begin{proof}
It is known that \textsc{EDP} is \NP-hard when restricted to the class $K_{3,N}$~\cite{FleszarMS18}.
The reduction takes an instance $\mathcal{I}$ of \textsc{EDP} on $K_{3,N}$ and replaces each edge $uv$ in the $K_{3,n}$ from $\mathcal{I}$ with the gadget depicted in Figure~\ref{fig:td-1-in-3} (left), obtaining a digraph $G'$.
We now complete the construction of the instance $\mathcal{I}'$ obtained from $\mathcal{I}$. First, the latency functions are set so that the capacity of each arc in $G'$ is set to one, and a flow of at most one costs nothing, that is, for each $e\in E(G')$ and $x\in \mathbb{N}$,
\[
\ell_e(x) = 
\begin{cases} 
	0 $ if $ x \leq 1 \\
	\infty $ otherwise$.
\end{cases}
\]
Then, for each terminal pair $(s,t)$ in $P$ in $\mathcal{I}$, we want to route $1$ agent from $s$ to $t$ in $G'$.
This completes the construction of $\mathcal{I}'$.
By construction, given two vertices $u,v\in V(G')$ connected by an arc gadget, it is only possible to route at most one agent from $u$ to $v$ or from $v$ to $u$ via this arc gadget.
Due to the capacity of each arc in $G'$ being $1$, we then obtain the desired equivalence.
That is, there is a solution to $\mathcal{I}$ if and only if there is a solution of cost $0$ for $\mathcal{I}'$.

To complete the proof, we observe that the obtained network contains a bounded-sized set $X$ of vertices (in our case, the part of size $3$ in the $K_{3,n}$) such that deleting $X$ decomposes the rest of the network into connected components each containing at most $y$ vertices (in our case $y=13$, as each such component consists of a vertex connected to $3$ copies of the gadget depicted in Figure~\ref{fig:td-1-in-3} (left). It is well-known (and also easy to see from their formal definitions~\cite{RobertsonS86,sparsity,CyganFKLMPPS15}) that such graphs have treewidth and treedepth upper-bounded by $|X|+y-1$ (i.e., $15$).
\end{proof}

\begin{figure}
\centering
\scalebox{0.87}{
\begin{tikzpicture}[scale=0.9]
\tikzstyle{every node}=[draw, shape=circle, minimum size=3pt,inner sep=0pt, fill=black]
\draw (0,0) node[label=left:$u$] (u){};
\draw (1.1,0) node (u'){}; 
\draw (2.2,1) node (b){}; 
\draw (2.2,-1) node (c){};
\draw (3.3,0) node (v'){}; 
\draw (4.4,0) node [label=right:$v$] (v){};
\draw[latex-latex] (u)--(u');
\draw[latex-latex] (v)--(v');
\draw[-latex] (u')--(b);
\draw[-latex] (v')--(b);
\draw[-latex] (c)--(v');
\draw[-latex] (c)--(u');
\draw[-latex] (b)--(c);
\draw (4.4,-1.9) node [color=white] (v){};
\end{tikzpicture}
}
\hfill
\vline
\hfill
\scalebox{0.8}{
\begin{tikzpicture}[scale=0.9]
\tikzstyle{every node}=[draw, shape=circle, minimum size=3pt,inner sep=0pt, fill=black]
\draw (-1,1.75) node[label=left:$x_1$] (x1){};
\draw (0,2.0) node (x1t){};
\draw (0,1.5) node (x1f){};
\draw[-latex] (x1) -- (x1t);
\draw[-latex] (x1) -- (x1f);

\draw (-1,0.25) node[label=left:$x_i$] (xi){};
\draw (0,0.5) node (xit){};
\draw (0,0) node (xif){};
\draw[-latex] (xi) -- (xit);
\draw[-latex] (xi) -- (xif);

\draw (-1,-1.25) node[label=left:$x_n$] (xn){};
\draw (0,-1) node (xnt){};
\draw (0,-1.5) node (xnf){};
\draw[-latex] (xn) -- (xnf);
\draw[-latex] (xn) -- (xnt);

\path (x1f) -- node[draw=none,fill=white]{\rvdots} (xit);
\path (xif) -- node[draw=none,fill=white]{\rvdots} (xnt);

\draw (2,2.3) node (c1t){};
\draw (2,1.8) node (c1f){};
\draw (2,0.5) node(cjt){};
\draw (2,0) node (cjf){};
\draw (2,0.5) node(cjt){};
\draw (2,-1.3) node (cmt){};
\draw (2,-1.8) node (cmf){};
\path (c1f) -- node[draw=none,fill=white]{\rvdots} (cjt);
\path (cjf) -- node[draw=none,fill=white]{\rvdots} (cmt);

\draw (3,2.05) node[label=right:$c_1$] (c1){};
\draw (3,0.25) node[label=right:$c_j$] (cj){};
\draw (3,-1.55) node[label=right:$c_m$] (cm){};

\draw[-latex] (cmf) -- (cm);
\draw[-latex] (cmt) -- (cm);
\draw[-latex] (cjf) -- (cj);
\draw[-latex] (cjt) -- (cj);
\draw[-latex] (c1f) -- (c1);
\draw[-latex] (c1t) -- (c1);

\draw[-latex] (x1t) -- (c1t);
\draw[-latex] (x1f) -- (c1f);
\draw[-latex] (x1t) -- (cjt);
\draw[-latex] (x1f) -- (cjf);
\draw[-latex] (x1t) -- (cmt);
\draw[-latex] (x1f) -- (cmf);

\draw[-latex] (xnt) -- (c1t);
\draw[-latex] (xnf) -- (c1f);
\draw[-latex] (xnt) -- (cjt);
\draw[-latex] (xnf) -- (cjf);
\draw[-latex] (xnt) -- (cmt);
\draw[-latex] (xnf) -- (cmf);

\draw[-latex] (xit) -- (c1t);
\draw[-latex] (xif) -- (c1f);
\draw[-latex] (xit) -- (cjt);
\draw[-latex] (xif) -- (cjf);
\draw[-latex] (xit) -- (cmt);
\draw[-latex] (xif) -- (cmf);
\end{tikzpicture}
}
\caption{The arc gadget replacing each undirected edge $uv$ in the reductions from the proofs of Thms.~\ref{hard-td} and~\ref{W1-tcw} (left), and the digraph constructed in the proof of Thm.~\ref{hard-dags} (right).}
\label{fig:td-1-in-3}
\end{figure}
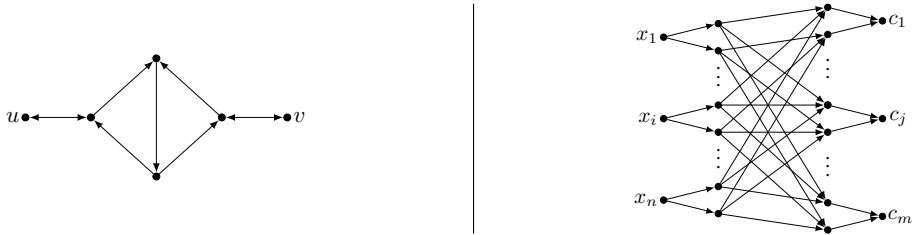

Theorem~\ref{hard-td} rules out using classic vertex-separator based parameters like treewidth, pathwidth or treedepth to solve SOAC, even when $\cmax=1$. All of these parameters are tied to the skeleton of the network and do not account for the orientations of the arcs. One could wonder whether structural parameters introduced specifically for directed graphs would be better suited for the task. The unifying feature of the most widely studied of these parameters~\cite{JohnsonRST01,hunter2008digraph,BerwangerDHKO12} is that they achieve constant values---typically $1$---on directed acyclic graphs (DAGs). Below, we show that bounded-capacity \problem\ is \NP-hard even on simple DAGs, ruling out the use of these directed network parameters.

\begin{theorem}
\label{hard-dags}
\problem\ is \NP-hard even when restricted to networks which are DAGs, have $\cmax=3$, and whose skeletons have maximum degree~$4$.
\end{theorem}

\begin{proof}
We reduce from a variant of SAT known as monotone cubic exact 1-in-3 SAT, which is known to be \NP-hard~\cite{PorschenSSW14,Schmidt2010d}.
Here, clauses have size three and must be satisfied by a single literal, each variable appears in three clauses, and all literals are positive.

Suppose we are given such an $n$-variate $m$-clause positive cubic $3$-CNF formula $\varphi$.
The instance of \problem is constructed as follows.
It contains the sets of vertices $X = \{x_1,\ldots,x_n\}$ and $C = \{c_1,\ldots,c_m\}$ that are identified with the variables and clauses appearing in $\varphi$, respectively.
We also add two more copies of $X$, namely $X^T = \{x^T_1,\ldots,x^T_n\}$ and $X^F = \{x^F_1,\ldots,x^F_n\}$, and, similarly, two more copies $C^T$ and $C^F$ of $C$. 
Now, $x_i$ has an arc to $x^T_i$ and one to $x^F_i$, both with latency $\ell(1) = \ell(2) = 1$ and $\ell(3) = 0$, and capacity $3$.
Whenever $x_i$ appears in the $j$-th clause of $\varphi$, then $x^T_i$ has an arc to $c^T_j$ and $x^F_i$ has an arc to $c^F_j$, both of latency $0$.
Further, $c^T_i$ has an arc to $c_i$ of capacity $1$ and $c^F_i$ has an arc to $c_i$ of capacity $2$, for all $i \in [m]$, and latency $0$ otherwise.
If $x_i$ appears in the $j$-th clause of $\varphi$, then there is an agent with origin $x_i$ and destination $c_j$. An illustration of the construction is provided in Figure \ref{fig:td-1-in-3} (right).

We claim that this instance of \problem{} has a solution of cost $0$ if and only if $\varphi$ is 1-in-3-satisfiable.
Any satisfying 1-in-3 assigment of $\varphi$ induces a routing of agents of cost $0$ by construction.
On the other hand, take any solution of cost $0$ to \problem.
This means, first, that all agents with origin $x_i$ are routed via the same choice of $x_i^T$ or $x_i^F$ by the respective latency functions of these arcs. In particular, we can derive an assignment $\pi$ to the variables from this information, and we claim that it is in fact a satisfying 1-in-3 assignment.
Towards showing that $\pi$ is satisfying, observe that, by the capacity constraints on the arc from $c^F_j$ to $c_j$, for every clause, at most $2$ variables can use the incoming arcs into $c^F_j$ from vertices $x^F_i$. Further, by the capacity constraints on the arc from $c^T_j$ to $c_j$, for every clause, at most $1$ variable can use the incoming arcs into $c^T_j$ from vertices $x^T_i$. Hence, exactly one of the agents must be routed via $c^T_j$, coming from a vertex $x^T_i$. Thus, every clause is satisfied by $\pi$ and $\pi$ is also a 1-in-3 assignment of $\varphi$.

Note that the produced instance has the desired properties from the statement of the theorem by the fact that $\varphi$ was assumed to be cubic.
\end{proof}

In fact, we can also show that our problem of interest remains intractable even on extremely simple networks which are DAGs (complementing Theorem~\ref{hard-vc}).

\begin{theorem}
\label{hard-tw-delta}
\problem\ is \NP-hard even when restricted to DAGs whose skeletons are planar graphs with treewidth and maximum degree both upper-bounded by $3$.
\end{theorem}

\begin{proof}
To prove the statement of the theorem, we present a reduction from \textsc{MUKS} to \problem.
Starting with an instance $\mathcal{I}$ of \textsc{MUKS} encoded in unary, we construct an instance $\mathcal{I}'$ of \problem{} as follows.

For each $\vec v_i \in S$ ($i\in [n]$) in $\mathcal{I}$, there is a source vertex $s_i$ in $\mathcal{I}'$.
For each entry $T_j\in \vec T$ ($j\in [d]$) in the target vector in $\mathcal{I}$, there is a terminal vertex $t_j$ in $\mathcal{I}'$.
The rest of $\mathcal{I}'$ is constructed as follows.
There are two vertices $h_i$ and $h'_i$ for each $i\in [n+1]$.
For each $i\in [n]$, there are the arcs $(h_i,h_{i+1})$, $(h'_i,h'_{i+1})$, $(s_i,h_i)$, and $(s_i,h'_i)$.
There is a directed almost complete binary tree (the level with the leaves may not be completely filled) whose arcs go from the root to the leaves with $h_{n+1}$ ($h'_{n+1}$, respectively) as the root and $t_1,\ldots,t_d$ as the leaves.
The two directed almost complete binary trees are symmetric with respect to the leaves.

For each $i\in [n]$ and $x\in \mathbb{N}$, let $\ell_{(s_i,h'_i)}(x) = 1/x$. 
For each $j\in [d]$, let $e_j$ be the arc incoming into $t_j$ from the directed almost complete binary tree with $h_{n+1}$ as the root, and, for each $x\in \mathbb{N}$, let
\[
\ell_{e_j}(x) = 
\begin{cases} 
	0 $ if $ x \leq T_j \\
	\infty $ otherwise$.
\end{cases}
\]
In other words, for each $j\in [d]$, the capacity of the arc $e_j$ is equal to $T_j$.
For each remaining arc $e$ for which a latency function has not yet been defined, let $\ell_{e}(x) = 0$ for each $x\in \mathbb{N}$.

For each $\vec v_i \in S$ ($i\in [n]$) and $j\in [d]$, we want to route $v_{i,j}$ agents from $s_i$ to $t_j$, where $v_{i,j}$ is the $j$-th entry of $\vec v_i$.
That is, there are $v_{i,j}$ copies of the tuple $(s_i, t_j)$ for each $i\in [n]$ and $j\in [d]$.
This completes the construction of $\mathcal{I}'$, which is clearly achieved in polynomial time.
See Figure~\ref{fig:tw_deg_red} for an illustration.

We now show that there is a solution comprised of at least $k$ vectors for $\mathcal{I}$ if and only if there is a solution of cost at most $n-k$ for $\mathcal{I}'$.
First, assume there is a solution comprised of at least $k$ vectors for $\mathcal{I}$.
For each of the vertices $s_i$ corresponding to one of these at least $k$ vectors in the solution for $\mathcal{I}$, route each of the agents at $s_i$ to their corresponding destination $t_j$ via the directed path that contains $h_i$.
Since these vectors are in the solution for $\mathcal{I}$, for each $j\in [d]$, the capacity of the arc $e_j$ has not been surpassed, and thus, so far, the cost of routing these agents is $0$ by the construction.
The remaining agents are all routed to the $t_j$'s via the directed path that contains $h'_{n+1}$. 
Since, for each $i\in [n]$ and $x\in \mathbb{N}$, we have that $\ell_e(x) = 1/x$ for any arc $e$ of the form $(s_i,h'_i)$, then $f(e) \ell_e(f(e)) = x \cdot 1/x = 1$, where $f(e) = x$.
Hence, the cost of a solution for $\mathcal{I}'$ is at most $n-k$ since there are at most $n-k$ such $s_i$'s where agents are routed to their destinations via $h'_{n+1}$, and the only arcs on these directed paths that have non-zero latency are the ones of the form $(s_i,h'_i)$.

Now, suppose we are given a solution for $\mathcal{I}'$ of total cost at most $n-k$.
By the property of the latency function $1/x$ for any arc of the form $(s_i,h'_i)$, this means that there at most $n-k$ vertices in $\{h'_1,\ldots,h'_n\}$ that had agents routed to them through their incoming arc from $s_i$.
Hence, for at least $k$ of the vertices in $s_1,\ldots,s_n$, all of their agents must have been routed to the $t_j$'s via the $e_j$'s.
Since the total cost of the solution for $\mathcal{I}'$ is at most $n-k$ (not $\infty$), then, for each $j\in [d]$, the capacity $T_j$ of the arc $e_j$ was not surpassed by this routing.
Hence, the sum of the vectors corresponding to these at least $k$ vertices in $s_1,\ldots,s_n$ must be at most $\vec T$.

It now remains to show that the digraph in $\mathcal{I}'$ has the desired properties from the theorem statement. 
It is trivial to see that the digraph in $\mathcal{I}'$ is a planar DAG, and that its skeleton has maximum degree~$3$.
To show that the treewidth of its skeleton is at most~$3$, we give a strategy for four cops to capture the robber in the equivalent cops and robbers definition of treewidth: a graph $G$ has treewidth $k$ if and only if the minimum number of cops that can ensure capturing the robber in $G$ is $k+1$~\cite{SeymourT93}.

The four cops first occupy $h_{n+1}$, $h'_{n+1}$, $s_n$, and $h_n$, respectively.
If the robber occupies an $s_i$, $h_i$ or $h'_i$ vertex, then it is easy to see that the cops can capture the robber.
Thus, assume that the robber territory is restricted to the two almost complete binary trees that are joined at the leaves minus their roots.
The first and second cops occupy $h_{n+1}$ and $h'_{n+1}$, respectively. The third and fourth cops move to a child of the first and second cops, respectively.
Let the first and third cop be one team of cops, and let the second and fourth cops be the other team of cops.
From this point on, the teams of cops play in an alternating fashion.
The first and third (the second and fourth, respectively) cops then apply the following well-known winning strategy for two cops in their respective tree~\cite{SeymourT93}: at each point of time the two cops will occupy two pairwise adjacent vertices in the tree, and in each round the cop occupying the vertex closer to the root will move to a child of the other cop in the direction of the robber if the robber is located in their tree. If the robber is located in the other tree, the strategy is the same but we attempt to catch the ``shadow'' of the robber in their subtree, where the robber's shadow is simply the vertex corresponding to the robber's position in the mirroring almost complete binary tree. 
\end{proof}

\begin{figure}
\begin{center}
\includegraphics[scale=0.64]{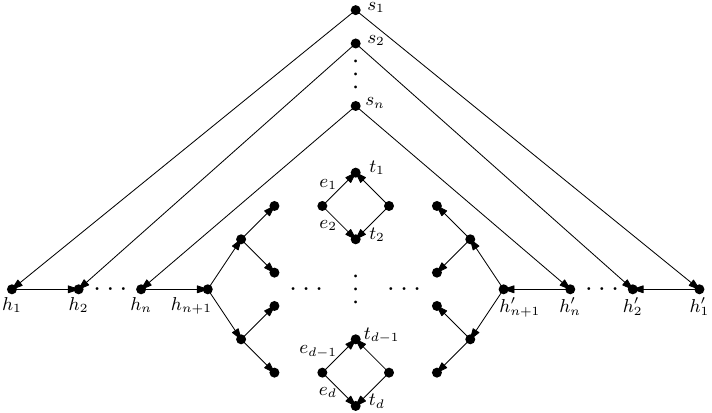}
\end{center}
\caption{The digraph constructed in the proof of Thm.~\ref{hard-tw-delta}.}
\label{fig:tw_deg_red}
\end{figure}

\section{A Fixed-Parameter Algorithm for \problem}

Recapitulating the results so far, in view of Theorem~\ref{hard-vc}, we are focusing our efforts on identifying structural properties of networks which would allow us to solve \problem\ on bounded-capacity networks. Theorem~\ref{hard-td} rules out tractability via standard vertex-separator based graph parameters, while Theorem~\ref{hard-dags} excludes the use of commonly studied directed variants of network parameters. A notable class of structural measures that has not been addressed yet by the obtained lower bounds are parameters that are tied to the existence of \emph{edge cuts} as opposed to \emph{vertex separators}.

Two ``baseline'' graph measures which allow for a structural decomposition along bounded-sized edge cuts are the \emph{feedback edge number} (\fen)~\cite{KoanaKNNZ21,FuchsleMNR22},
and the combined parameter of \emph{treewidth plus maximum degree} (\twd)~\cite{OrdyniakS13,GozupekOPSS17}.  Unfortunately, both of these measures place very strong restrictions on the network: the former is only small on networks whose skeletons are trees plus a small number of arcs, while the latter cannot be used on networks with high-degree nodes. 

As our main algorithmic contribution, we establish fixed-parameter tractability of bounded-capacity \problem\ with respect to a recently introduced edge-cut based parameter called \emph{slim treecut width} (\stcw)~\cite{GanianK22}. Crucially, this result immediately implies and generalizes fixed-parameter tractability with respect to either \fen\ or \twd, while also circumventing both of the aforementioned shortcomings: networks with skeletons of bounded slim treecut width can have high degree as well as be significantly more complex than just a tree. 

\begin{theorem}
\label{thm:stcwfpt}
\problem\ is fixed-parameter tractable when parameterized by the slim treecut width of the skeleton $\underline G$ of the input digraph $G$ plus the maximum capacity $\cmax$.
\end{theorem}

\begin{proof}
We begin by invoking Lemma~\ref{lem:computestcw} to compute a pair $(H,T)$ such that $H$ is a supergraph of $\underline G$ and $(H,T)$ has edge-cut width $k\leq 3\cdot (\kappa+1)^2$, where $\kappa$ is the slim treecut width of $\underline G$. Our algorithm is based on a leaf-to-root dynamic programming procedure that traverses $T$ while storing certain (carefully defined and bounded-sized) records about the part of $T$ that has been processed so far. To this end, it will be useful to assume that $T$ is rooted, and thus, we mark an arbitrary leaf of $T$ as the root and denote it~$r$.

Before we define the records used in the dynamic program, we will need some terminology. For a vertex $v\in V(H)$ with a child $w\in V(H)$, we say that $w$ is a \emph{simple} child of $v$ if $v$ and $w$ belong to different connected components of $H-vw$; otherwise, we say that $w$ is a \emph{complex} child of $v$. Observe that, while the number of simple children of $v$ is not bounded by $k$ (none of the subtrees rooted at simple children contain any edges in $E_{\loc}^{H,T}(v)$), the number of complex children of $v$ is upper-bounded by $2k$ (each subtree rooted at a complex child contains an endpoint of at least one edge in $E_{\loc}^{H,T}(v)$). Furthermore, we use $\forgottenG_v$ to denote the sub-digraph of $G$ induced on the vertices that are descendants of $v$ (including $v$ itself), and $\partial_v$ to denote those arcs of $G$ which have precisely one endpoint in $\forgottenG_v$; recall that $|\partial_v|\leq 2(k+1)$. An agent $a_i$ is \emph{outgoing} for $v$ if $s_i\in \forgottenG_v$ and, at the same time, $t_i\not \in \forgottenG_v$. Similarly, an agent $a_i$ is \emph{incoming} for $v$ if $s_i\notin \forgottenG_v$ and, at the same time, $t_i \in \forgottenG_v$.

We are now ready to formalize the dynamic programming records used in our algorithm. A \emph{snapshot} at a vertex $v$ is a tuple of the form $(\Sout,\Sin,\D,\R)$, where:

	\begin{itemize}[topsep=0pt]
		\item $\Sout$ is a mapping from the set of all outgoing agents for $v$ to arcs in $\partial_v$ which are outgoing from $\forgottenG_v$;
		\item $\Sin$ is a mapping from the set of all incoming agents for $v$ to arcs in $\partial_v$  which are incoming to $\forgottenG_v$;
		\item $\D$ is a multiset of pairs $(e,f)$ such that $e$ is an incoming arc into $\forgottenG_v$, $f$ is an outgoing arc from $\forgottenG_v$, and $ef$ is not a $2$-cycle;
		\item $\R$ is a multiset of pairs $(e,f)$ such that $e$ is an outgoing arc from $\forgottenG_v$, $f$ is an incoming arc into $\forgottenG_v$, and $ef$ is not a $2$-cycle;
		\item each arc in $\partial_v$ may only appear in at most $\cmax$ tuples over all of the entries in $\Sout, \Sin, \D, \R$.
	\end{itemize}

Before proceeding, it will be useful to obtain an upper-bound on the total number of possible snapshots for an arbitrary vertex $v$. First, we observe that if the number of outgoing agents for $v$ exceeds $\cmax \cdot(k+1)$, then every flow assignment in the considered instance must necessarily exceed the capacity $\cmax$ for at least one arc in $G$; in other words, such an instance can be immediately recognized as a NO-instance. The same also holds for the number of incoming agents for $v$. Hence, we can restrict our attention to the case where these simple checks do not fail, and use this to obtain a bound on the total number of snapshots at $v$. In particular, there are at most $(2k+2)^{\cmax \cdot(k+1)}$ possible choices for each of $\Sout$ and $\Sin$ (there are at most $2k+2$ choices of which element of $\partial_v$ to map each of the at most $\cmax \cdot(k+1)$ many agents to), and also at most $(\cmax+1)^{4(k+1)^2}$ possible choices for $\D$ and $\R$ each (for each of the at most $4(k+1)^2$ possible unordered pairs of arcs from $\partial_v$, there are $(\cmax+1)$ possible choices of how many times it occurs in $\D$ and $\R$). We conclude that the number of snapshots at $v$ can be upper-bounded by $(\cmax+k)^{\bigoh(\cmax k^2)}$, and let $\snap(v)$ denote the set of all possible snapshots at $v$.

We can now formalize the syntax of the record at $v$, denoted $\Rec(v)$, as a mapping from $\snap(v)$ to $\mathbb{R}_{\geq 0}\cup \{\infty\}$. As for the semantics, $\Rec(v)$ will capture the minimum cost required to (1) route the outgoing and incoming agents to the designated arcs in $\Sout$ and $\Sin$, while assuming that (2) some (unidentified and arbitrary) agents will use arcs of $\partial_v$ to enter and then exit $\forgottenG_v$ via the arcs designated in $\D$, and that (3) some (unidentified and arbitrary) agents will use arcs of $\partial_v$ to exit and then return to $\forgottenG_v$ via the arcs designated in~$\R$. 
Formally, $\Rec(v)$ maps each snapshot $\Upsilon=(\Sout,\Sin,\D,\R)$ to the minimum cost of a flow assignment $F$ in the subinstance $\mathcal{I}_\Upsilon$ induced on the vertices of $\forgottenG_v$ plus the arcs\footnote{We remark that since only one of the endpoints of $\partial_v$ lies in $\forgottenG_v$, these are not formally arcs in $\forgottenG_v$.} of $\partial_v$ with the following properties:
\begin{enumerate}
\item for each pair $(ab,cd)$ of arcs in the multiset $\D$ such that $b,c\in \forgottenG_v$, we add a new marker agent into $I_\Upsilon$ which starts at $a$ and ends at $d$;
\item for each pair $(ab,cd)$ of arcs in the multiset $\R$ such that $a,d\in \forgottenG_v$, we add a new arc $bc$ into $\mathcal{I}_\Upsilon$ and set $\ell_{bc}:=\{1\mapsto 0\}\cup \{i\mapsto \infty~|~i>1\}$.
\item for each outgoing agent $a_i$, $F$ contains a path from $s_i$ into the arc $\Sout(a_i)$;
\item for each incoming agent $a_i$, $F$ contains a path from (and including) the arc $\Sin(a_i)$ to $t_i$.
\end{enumerate}

If no flow with these properties exists, we simply set $\Rec(v)(\Upsilon):=\infty$. This completes the formal definition of the records $\Rec(v)$. Observe that, since $\forgottenG_r=G$ and $\partial_r=\emptyset$, the only snapshot at $r$ is $\Upsilon_r=(\{\emptyset\},\{\emptyset\},\emptyset,\emptyset)$ and $\mathcal{I}_{\Upsilon_r}$ is precisely the input instance to our problem. Hence, if we successfully compute the record $\Rec(r)$ for the root vertex $r$, then $\Rec(r)(\Upsilon_r)$ must be equal to the minimum cost of a flow assignment $F$ in the input instance. Thus, to conclude the proof it remains to show how to compute the records at each vertex in a leaf-to-root fashion. 

Towards this task, let us first consider the computation of $\Rec(v)$ for a leaf $v$ in $T$. Here, we observe that each of the constructed instances $\mathcal{I}_\Upsilon$ consists of the at most $2k+2$ arcs in $\partial_v$, and the at most $(k+1)\cdot \cmax$ arcs obtained from $\R$. The number of paths in such an instance is upper-bounded by $2^{\bigoh(k\cdot \cmax)}$, and hence, the minimum cost of a flow assignment $F$ in $\mathcal{I}_\Upsilon$ can be computed by enumerating all flow assignments in time at most $\cmax^{2^{\bigoh(k\cdot \cmax)}}$.

The core of the dynamic program lies in the computation of $\Rec(v)$ for a non-leaf vertex $v$. We do so by considering each snapshot $\Upsilon=(\Sout,\Sin,\D,\R)$ at $v$ independently, and computing the minimum cost of a flow assignment $F$ in the subinstance $\mathcal{I}_\Upsilon$ as follows. For each simple child $w$ of $v$, we observe that there is only a single snapshot $\Psi_w=(\Sout^w,\Sin^w,\{\emptyset\},\{\emptyset\})$ of $w$ where $\Sout^w$ maps all outgoing agents for $w$ and $\Sin^w$ maps all incoming agents for $w$ to the respective unique arcs in $\partial_w$. This corresponds to the fact that since a flow assignment is inherently loopless, there is only a unique way it may pass through the arcs in $\partial_w$. Let us define the \emph{base cost} of $v$, denoted $b(v)$, as $\sum_{w\text{ is a simple child of v}}\Rec(w)(\Psi_w)$; intuitively, $b(v)$ captures the minimum cost required by a flow assignment in all simple children of $v$. 

Next, we construct the instance $\mathcal{I}_\Upsilon$ for which we need to compute the minimum cost of a flow assignment. Essentially, our aim is to compute this cost via brute-forcing over all possible flow assignments, but at first glance this seems infeasible since the size of $\mathcal{I}_\Upsilon$ is not bounded by our parameters. The core insight we use to overcome this is that even though $\mathcal{I}_\Upsilon$ could be large, all but only a parameter-bounded number of interactions have already been taken into account in the records of the children of $v$.
Formally, we make use of this by constructing a ``kernelized'' instance $\mathcal{I}^+_\Upsilon$ as follows:

First, for each simple child $w$ of $v$, delete the whole subtree rooted at $w$, whereas, for each outgoing agent $a_i$ at $w$, we place $s_i$ on $v$, and, for each incoming agent $a_i$ at $w$, we place $t_i$ on $v$.
Then, for each complex child $w$ of $v$, we
\begin{enumerate}
\item delete every vertex in $\forgottenG_w$ except for the endpoints of $\partial_w$;
\item for each pair of endpoints $a,b\in V(\forgottenG_w)$ of distinct arcs in $\partial_w$, add bi-directional ``marker'' arcs between $a$ and $b$;
\item create a vertex $w_\text{out}$, add a directed arc from $w_\text{out}$ to every endpoint of an arc $\partial_w$ in $V(\forgottenG_w)$, and for outgoing agent $a_i$ at $w$ we place $s_i$ on $w_\text{out}$;
\item create a vertex $w_\text{in}$, add a directed arc from every endpoint of an arc $\partial_w$ in $V(\forgottenG_w)$ to $w_\text{in}$, and for an incoming agent $a_i$ at $w$ we place $t_i$ on $w_\text{in}$.
\end{enumerate}

We remark that the newly created arcs are not associated with a latency function (or alternatively may be assumed to be associated with the degenerate latency function $\mathbb{N}\rightarrow \{0\}$); these arcs are merely used as markers to point us to which of the snapshots should be used for each complex child of $v$. Crucially, the number of vertices in $\mathcal{I}^+_\Upsilon$ is upper-bounded by $(k\cdot 2k)+1\leq \bigoh(k^2)$. The number of paths in such an instance is upper-bounded by $k^{\bigoh(k^2)}$, and hence, the total number of all possible flow assignments can be upper-bounded by $\cmax^{k^{\bigoh(k^2)}}$. We proceed by enumerating the set of all potential flow assignments, and, for each such potential flow assignment, we check whether it is indeed a flow assignment for the agents specified in $\mathcal{I}^+_\Upsilon$. If this check succeds, we employ a further technical check to ensure that each path in $P$ only contains a marker arc $e$ in a complex child $w$ if $e$ is immediately preceded and also succeeded by an arc in~$\partial_w$. 

For each flow assignment in $\mathcal{I}^+_\Upsilon$ that passes the above checks and, for each complex child $w$ of $v$, we observe that $F$ restricted to the arcs with at most one endpoint in $w$ fully determines (1) the first arc in $\partial_w$ used by an outgoing agent for $w$, (2) the last arc in $\partial_w$ used by an incoming agent for $w$, (3) which pairs of arcs in $\partial_w$ are used by paths to leave and subsequently re-enter $\forgottenG_w$, and (4) which pairs of arcs in $\partial_w$ are used by paths to enter and subsequently leave $\forgottenG_w$. This information hence identifies a unique snapshot $\Psi^w_F$ of $w$. We define the cost of $F$ as $b(v)+\sum_{w\text{ is a complex child of }v}\Rec(w)(\Psi^w_F)$. Finally, we set $\Rec(v)(\Upsilon)$ to be the minimum cost of a flow $F$ in $\mathcal{I}^+_\Upsilon$ which satisfies the conditions stipulated in this paragraph; this concludes the description of the algorithm.

Apart from the time required to compute $(H,T)$, which is upper-bounded by $2^{\kappa^{\bigoh(1)}}\cdot n^4$,
the running time of the algorithm can be upper-bounded by the number of vertices in the input digraph times the cost of processing each vertex. The latter is dominated by $\cmax^{k^{\bigoh(k^2)}}$, i.e., $\cmax^{\kappa^{\bigoh(\kappa^4)}}$. We remark that as with essentially all width-based dynamic programming routines, the algorithm can be made constructive by performing a subsequent top-to-bottom computation in order to compute a specific flow assignment with the claimed flow as a witness. To argue correctness, it hence remains to argue that if the input instance admits a flow assignment $Q$ of minimum cost, say $p$, then the algorithm will compute a flow assignment with the same cost $p$. Towards this goal, we observe that at each vertex $v$ considered in the leaf-to-root pass made by the dynamic program, $Q$ will correspond to a unique snapshot $\Upsilon$ of $v$. At each leaf $v$ of $T$, $\Rec(v)(\Upsilon)$ must be equal to the cost incurred by $Q$ on the arcs in $\partial_v$ due to the nature of our brute-force computation of the records for trees and the optimality of $Q$. Moreover, for each non-leaf node $v$, it holds that as long as we have correctly computed the records for each of its children, the traversal of $Q$ via the arcs in $\partial_v$ and $\partial_w$ for each child of $w$ identifies a unique valid flow assignment $F$ in $\mathcal{I}^+_\Upsilon$, and this in turn defines a snapshot $\Psi^w_F$ for each child $w$ of $v$. In that case, however, $Q$ must indeed incur a cost of $b(v)+\sum_{w\text{ is a complex child of }v}\Rec(w)(\Psi^w_F)$ over all arcs in $\mathcal{I}_\Upsilon$, as desired.
\end{proof}

We observe that the parameterization by $\cmax$ cannot be dropped from Theorem~\ref{thm:stcwfpt} in view of the lower bound in Theorem~\ref{hard-tw-delta}; indeed, every network of bounded treewidth and maximum degree also has bounded slim treecut width (Lemma~\ref{lem:0tcw}). We conclude by turning our attention to whether Theorem~\ref{thm:stcwfpt} could be generalized to use the better-known \emph{treecut width} parameter instead of the slim variant used in that algorithm. Treecut width is a structural graph parameter that also guarantees the decomposability of instances via bounded-sized edge cuts, and has previously been used to establish tractability for several \NP-hard problems~\cite{GanianKS22}. Crucially, it forms an intermediary between the slim treecut width (which suffices for fixed-parameter tractability) and treewidth (for which \problem\ remains intractable, even when restricted to bounded-capacity instances). We show that Theorem~\ref{thm:stcwfpt} is tight in the sense that fixed-parameter tractability cannot be lifted to treecut width. 

\begin{theorem}
\label{W1-tcw}
\problem\ is \W\textup{[1]}-hard when parameterized by the treecut width of the skeleton of the input network, even when restricted to networks with $\cmax=1$.
\end{theorem}

\begin{proof}
\textsc{EDP} is known to be \W[1]-hard parameterized by the treecut width of the graph, even when restricted to the class $K_{3,N}$~\cite{GanianO21}. To establish the theorem, we use the same reduction as in the proof of Theorem~\ref{hard-td}, and prove that the resulting instances also have bounded treecut width. Let us consider an output network $G'$ obtained from that reduction.

Since the input graph $G$ has bounded treecut width, it suffices to show that replacing each of the edges in $G$ by the arc gadget depicted in Figure~\ref{fig:td-1-in-3} (left) does not increase the treecut width by too much.
We actually prove that $\tcw(\underline{G'})\leq \tcw(G)$ since we can assume that $\tcw(G)$ can be assumed to be at least $5$ in the reduction from $\mathcal{I}$ to $\mathcal{I}'$.

Consider a treecut decomposition $(T,\mathcal{X})$ of $G$ of width $\tcw(G)\geq 5$.
From $(T,\mathcal{X})$, we will construct a treecut decomposition $(T',\mathcal{X}')$ of $\underline{G'}$ of width at most $\tcw(G)$.
Initially, set $T'=T$ and $\mathcal{X}'=\mathcal{X}$.
Consider any edge $uv\in E(G)$ and let $U\subset V(G')$ be the four vertices in the arc gadget connecting $u$ and $v$ in $G'$.
If $u$ and $v$ are in $X_t$ for some node $t$ of $T$, then, in $T'$, we add a child $t'$ of $t$ such that $X_{t'}=U$.
If $u$ and $v$ are in $X_{t_1}$ and $X_{t_2}$ for two different nodes $t_1$ and $t_2$ of $T$, and, without loss of generality, the depth of $t_1$ is at least the depth of $t_2$ in $T$, then, in $T'$, we add a child $t'$ of $t_1$ such that $X_{t'}=U$.
In both cases, $\mathcal{X}'$ is updated accordingly.

Let $(T',\mathcal{X}')$ be the treecut decomposition of $\underline{G'}$ obtained from $(T,\mathcal{X})$ by applying the above procedure for each edge in $E(G)$.
For each child $t'$ added to $T$ in the process of obtaining $T'$, there are exactly four vertices in $X_{t'}$ and, in $G'$, they are only adjacent to each other and two vertices $u,v\in X_t$ ($u\in X_{t_1}$ and $v\in X_{t_2}$, respectively) in the first (second, respectively) case. 
Thus, $\tor(t')\leq 5$, $\adh(t')\leq 2$, and the torso-size of any other node in $T'-t'$ in $(T',\mathcal{X}')$ is the same as in $(T,\mathcal{X})$.
Indeed, this may only not be obvious in the case of the torso-size of $t$ in $T'$, but the vertex $z$, resulting from consolidating each vertex of $X_{t'}$ into a single vertex, only has degree~$2$ in $\underline{G'}$, and thus, $z$ will be suppressed and an edge will be added between $u$ and $v$ (which already exists in $G$) when forming the 3-center of $(H_t,\mathcal{X}')$, where $H_t$ is the torso of $(T',\mathcal{X}')$.
Note that the depth of $t'$ is always strictly greater than the depth of $t$ ($t_1$ and $t_2$, respectively) in the first (second, respectively) case.
Thus, the adhesion of any other node in $T'-t'$ in $(T',\mathcal{X}')$ is the same as in $(T,\mathcal{X})$.
Hence, $\tcw(\underline{G'})\leq \tcw(G)$. 
\end{proof}

\section{The Min-Max Atomic Congestion Problem}

In this section, we turn our attention to the more general setting where instead of requiring \emph{all} agents to be routed to their destinations, we allow for some agents to
remain unrouted.
In essence, this asks for a flow assignment that counterbalances the number of agents that reach their destinations with the total cost. As is usual in complexity-theoretic analysis, we state this task as a decision problem where we consider specific bounds on both the cost and the number of agents which need not be routed. For the purposes of this section, we formally extend the notion of \emph{flow assignment} (defined in Section~\ref{sec:prelims}) to a mapping from a subset of agents to paths.

\defproblem{Min-Max System Opt.\ Atomic Congestion (\mproblem)}{A digraph $G=(V,E)$, positive integers $\lambda$ and $\alpha$, a set $A=\{a_1,\dots,a_m\}$ of agents where each agent $a_i$ is associated with a tuple $(s_i,t_i)\in V^2$, 
and for each arc $e\in E$, a latency function $\ell_e: [m] \rightarrow \mathbb{R}_{\geq 0} \cup \{\infty\}$.}{Is there a flow assignment $F$ routing at least $m-\alpha$ agents such that $\sum_{e\in E}f_F(e)\cdot \ell_e(f_F(e)) \leq \lambda$?}

We begin by noting that, in spite of the seemingly minor difference between the two problems, \mproblem\ is much more challenging than \problem\ from a structural point of view. Indeed, on one hand, if we set $\alpha=0$, the min-max variant becomes equivalent to \problem. On the other hand, without this restriction, it can be observed that \mproblem\ remains \NP-hard even on networks with a maximum capacity of $1$ whose skeletons are trees. This follows by an immediate reduction from the \textsc{Maximum Arc Disjoint Paths} problem, which has been shown to be \NP-hard in precisely this setting~\cite{ErlebachJansen01}; the reduction simply replaces each path with an agent and sets the latency function for each arc in the same way as in the proof of Theorem~\ref{hard-td}. 

As the final contribution of this article, we show that \mproblem\ is fixed-parameter tractable when parameterized simultaneously by $\cmax$, $\alpha$, and the combined parameter of treewidth and maximum degree.

\begin{theorem}
\label{thm:msoacfpt}
\mproblem\ is fixed-parameter tractable when parameterized by the treewidth and maximum degree of the skeleton $\underline{G}$ of the input digraph $G$ plus the maximum capacity $\cmax$ and $\alpha$.
\end{theorem}

\begin{proof}
On a high level, the algorithm follows the same overall approach as the one employed in the proof of Theorem~\ref{thm:stcwfpt}; however, the dynamic programming steps and records required here are different (and more complicated). We begin by invoking Lemma~\ref{lem:0tcw} to compute a spanning tree $T$ over a supergraph $H$ of $\underline G$ such that there is an integer $k\in \bigoh(d^2w^2)$ which is an upper-bound for both the edge-cut width of $(H,T)$ and the maximum degree of $T$. As in the proof of Theorem~\ref{thm:stcwfpt}, we mark an arbitrary leaf of $T$ as the root and denote it~$r$.

Recalling the notions of outgoing and incoming agents for a node $v$ from Theorem~\ref{thm:stcwfpt}, it will be useful to observe that if $T$ contains a node $v$ such that the number of outgoing or incoming agents in $\forgottenG_v$ exceeds $\alpha+\cmax\cdot (k+1)$, then we can immediately recognize the input as a NO-instance; this is because at most $\cmax\cdot (k+1)$ paths can traverse the edges forming the cut between $\forgottenG_v$ and the rest of the network. Hence, we hereinafter assume that the number of incoming and outgoing agents are both bounded by $\alpha+\cmax\cdot (k+1)$, and proceed to formalizing the dynamic programming records used in our algorithm.
We define a similar notion of a \emph{snapshot} at a vertex $v\in V(H)$, albeit here snapshots are tuples of the form $(\Aout,\Ain,\alpha',\Sout,\Sin,\D,\R)$ where:
	\begin{itemize}[topsep=0pt]
\item $\Aout$ is a subset of the set of all outgoing agents for $v$;
\item $\Ain$ is a subset of the set of all incoming agents for $v$;
\item $\alpha'\leq \alpha$ is a non-negative integer;
		\item $\Sout$ is a mapping from $\Aout$ to arcs in $\partial_v$ which are outgoing from $\forgottenG_v$;
		\item $\Sin$ is a mapping from $\Ain$ to arcs in $\partial_v$ which are incoming to $\forgottenG_v$;
		\item $\D$ is a multiset of pairs $(e,f)$ such that $e$ is an incoming arc into $\forgottenG_v$, $f$ is an outgoing arc from $\forgottenG_v$, and $ef$ is not a $2$-cycle;
		\item $\R$ is a multiset of pairs $(e,f)$ such that $e$ is an outgoing arc from $\forgottenG_v$, $f$ is an incoming arc into $\forgottenG_v$, and $ef$ is not a $2$-cycle;
		\item each arc in $\partial_v$ may only appear in at most $\cmax$ tuples over all of the entries in $\Sout, \Sin, \D, \R$.
	\end{itemize}

Comparing the snapshots defined above with those used in the proof of Theorem~\ref{thm:stcwfpt}, the only difference is that we (1) use $\alpha'$ to keep track of how many agents have been left unrouted in $\forgottenG_v$, and (2) we use $\Aout$ and $\Ain$ to keep track of the identities of outgoing and incoming agents which are not being routed. Hence, we can upper-bound the total number of possible snapshots for an arbitrary vertex $v$ as a product of the number of possible snapshots in the proof of Theorem~\ref{thm:stcwfpt} (i.e., $(\cmax+k)^{\bigoh(\cmax k^2)}$) times the number of choices for $\alpha'$ (i.e., $\alpha$) times the number of choices for $\Ain$ and $\Aout$ (i.e., $2^{\bigoh(\alpha+\cmax\cdot k)}$). Altogether, this yields an upper-bound of $(\cmax+k)^{\bigoh(\alpha \cdot \cmax \cdot k^2)}$. 

The syntax of the record at $v$, denoted $\Rec(v)$ is a mapping from $\snap(v)$ to $\mathbb{R}_{\geq 0}\cup \{\infty\}$. For the semantics, $\Rec(v)$ will capture the minimum cost required to (1) route the outgoing and incoming agents specified in $\Aout$ and $\Ain$ to the designated arcs in $\Sout$ and $\Sin$ while leaving the remaining outgoing and incoming agents unrouted, under the following three assumptions: (2) some (unidentified and arbitrary) agents will use arcs of $\partial_v$ to enter and then exit $\forgottenG_v$ precisely via the arcs designated in $\D$, (3) some (unidentified and arbitrary) agents will use arcs of $\partial_v$ to exit and then return to $\forgottenG_v$ via the arcs designated in $\R$, and (4) precisely $\alpha'$ many agents with at least one endpoint in $\forgottenG_v$ are not routed to their destinations. 

Formalizing the above, $\Rec(v)$ maps each snapshot $\Upsilon$ to the minimum cost of a flow assignment $F$ in the subinstance $\mathcal{I}_\Upsilon$ induced on the vertices of $\forgottenG_v$ plus the arcs\footnote{As before, only one of the endpoints of $\partial_v$ lies in $\forgottenG_v$, and hence, these are not formally arcs in $\forgottenG_v$.} of $\partial_v$ with the following properties:
\begin{enumerate}
\item for each pair $(ab,cd)$ of arcs in the multiset $\D$ such that $b,c\in \forgottenG_v$, we add a new marker agent into $I_\Upsilon$ which starts at $a$ and ends at $d$;
\item for each pair $(ab,cd)$ of arcs in the multiset $\R$ such that $a,d\in \forgottenG_v$, we add a new arc $bc$ into $\mathcal{I}_\Upsilon$ and set $\ell_{bc}:=\{1\mapsto 0\}\cup \{i\mapsto \infty~|~i>1\}$.
\item for each outgoing agent $a_i$ in $\Aout$, $F$ contains a path from $s_i$ into the arc $\Sout(a_i)$;
\item for each incoming agent $a_i$ in $\Ain$, $F$ contains a path from (and including) the arc $\Sin(a_i)$ to $t_i$;
\item $F$ routes all agents with at least one endpoint in $F$ except for $\alpha'$ many. Moreover, each outgoing (incoming, respectively) agent that is not in $\Aout$ ($\Ain$, respectively) is not routed by $F$.
\end{enumerate}

We remark that if no flow with these properties exists, we simply set $\Rec(v)(\Upsilon):=\infty$. This completes the formal definition of the records $\Rec(v)$. Observe that since $\forgottenG_r=G$ and $\partial_r=\emptyset$, the only snapshots at $r$ are of the form $\mathcal{I}_{\alpha'}=(\emptyset,\emptyset,\alpha',\{\emptyset\},\{\emptyset\},\emptyset,\emptyset)$, and each of these snapshots is mapped by $\Rec(r)$ to the minimum cost of a flow assignment in the input network which routes all but precisely $\alpha'$ agents to their destinations. Hence, $\min_{\alpha'\in [\alpha]}\Rec(r)(\mathcal{I}_{\alpha'})$ is the minimum cost of a flow assignment $F$ which routes at least $m-\alpha$ agents, and thus, to conclude the proof
it remains to show how to compute the records at each vertex in a leaf-to-root fashion.

Towards this task, let us first consider the computation of $\Rec(v)$ for a leaf $v$ in $T$. This step is based on exhaustive branching and is entirely analogous to the one employed in the proof of Theorem~\ref{thm:stcwfpt}, but with the distinction that we also branch on $\Ain$ and $\Aout$ (it is worth noting that, for leaves, each choice of $\Ain$ and $\Aout$ fully determines the value of $\alpha'$). More precisely, we observe that each of the constructed instances $\mathcal{I}_\Upsilon$ consists of the at most $2k+2$ arcs in $\partial_v$, the at most $(k+1)\cdot \cmax$ arcs obtained from $\R$, and the at most $(2k+2)\cdot \cmax$ agents that are in $\Ain$ or $\Aout$ for $v$. The number of paths in such an instance is upper-bounded by $2^{\bigoh(k\cdot \cmax)}$. Hence, the minimum cost of a flow assignment $F$ in $\mathcal{I}_\Upsilon$ that routes all the agents in $\Ain$ and $\Aout$ can be computed by enumerating all flow assignments in time at most $\cmax^{2^{\bigoh(k\cdot \cmax)}}$.

The core of the dynamic program lies in the computation of $\Rec(v)$ for a non-leaf vertex $v$. We do so by considering each snapshot $\Upsilon=(\Aout,\Ain,\alpha',\Sout,\Sin,\D,\R)$ at $v$ independently, and computing the minimum cost of a flow assignment $F$ in the subinstance $\mathcal{I}_\Upsilon$ as follows. Unlike in Theorem~\ref{thm:stcwfpt}, we directly construct the instance $\mathcal{I}_\Upsilon$ for which we need to compute the minimum cost of a flow assignment, as here we do not distinguish between ``simple'' and ``complex'' children, followed by a construction of the same ``kernelized'' instance $\mathcal{I}^+_\Upsilon$ as earlier (whereas this time, we treat \emph{every} child of $v$ as complex).  
In particular, for each child $w$ of $v$, we 
\begin{enumerate}
\item delete every vertex in $\forgottenG_w$ except for the endpoints of $\partial_w$;
\item for each pair of endpoints $a,b\in V(\forgottenG_w)$ of distinct arcs in $\partial_w$, add bi-directional ``marker'' arcs between $a$ and $b$;
\item create a vertex $w_\text{out}$, add a directed arc from $w_\text{out}$ to every endpoint of an arc $\partial_w$ in $V(\forgottenG_w)$, and for outgoing agent $a_i$ at $w$ we place $s_i$ on $w_\text{out}$;
\item create a vertex $w_\text{in}$, add a directed arc from every endpoint of an arc $\partial_w$ in $V(\forgottenG_w)$ to $w_\text{in}$, and for an incoming agent $a_i$ at $w$ we place $t_i$ on $w_\text{in}$.
\end{enumerate}

As before, the newly created arcs are not associated with a latency function (or alternatively may be assumed to be associated with the degenerate latency function $\mathbb{N}\rightarrow \{0\}$), and we observe that the number of vertices in $\mathcal{I}^+_\Upsilon$ is upper-bounded by $(k\cdot 2k)+1\leq \bigoh(k^2)$. The number of paths in such an instance is upper-bounded by $k^{\bigoh(k^2)}$, and hence, the total number of all possible flow assignments (regardless of the subset of routed agents) can be upper-bounded by $\cmax^{k^{\bigoh(k^2)}}$. We proceed by branching over each choice of a potential flow assignment $F$ from this set of all potential flow assignments. 

Next, we deal with the fact that $F$ need not be a flow assignment of all agents by performing an additional bounded branching step to determine which of the outgoing and incoming agents from each child of $v$ is actually routed by $F$. To this end, let an agent be \emph{important} for $v$ if it is an incoming or outgoing agent for at least one child of $v$, but is neither an incoming nor outgoing agent for $v$; in other words, important agents are those which are explicitly tracked by the snapshots of the children of $v$, but are no longer tracked by snapshots of $v$ itself. The number of important agents is upper-bounded by $k\cdot (\alpha+\cmax\cdot k)$, and we branch over each subset $\Imp$ of important children for $v$. Finally, we branch over all mappings $\beta$ from the children of $v$ to $[\alpha]$.

For each fixed $F$, $\Imp$, and $\beta$, we check that $F$ routes the agents in $\Imp$ to their final destinations and the agents in $\Ain$ and $\Aout$ to their assigned arcs as per $\Sin$ and $\Sout$, respectively. As previously in the proof of Theorem~\ref{thm:stcwfpt}, we also perform a further technical check to ensure that each path in $P$ only contains a marker arc $e$ in a child $w$ if $e$ is immediately preceded and also succeeded by an arc in $\partial_w$. If these checks succeed, we view $F$ as a ``projection'' of a flow assignment in $\mathcal{I}_\Upsilon$ onto $\mathcal{I}^+_\Upsilon$. In particular, $F$ restricted to the arcs with at most one endpoint in a child $w$ of $v$ fully determines (1) the first arc in $\partial_w$ used by an outgoing agent for $w$, (2) the last arc in $\partial_w$ used by an incoming agent for $w$, (3) which pairs of arcs in $\partial_w$ are used by paths to leave and subsequently re-enter $\forgottenG_w$, and (4) which pairs of arcs in $\partial_w$ are used by paths to enter and subsequently leave $\forgottenG_w$. This information, combined with information about which outgoing and incoming agents for $w$ are actually routed by the flow (specified in $\Imp$) and information about how many total agents with an endpoint in $\forgottenG_w$ are not routed by the flow (specified in $\beta(w)$), hence fully identifies a unique snapshot $\Psi^w_F$ of $w$. We define the cost of $F$ as $b(v)+\sum_{w\text{ is a child of }v}\Rec(w)(\Psi^w_F)$. Finally, we set $\Rec(v)(\Upsilon)$ to be the minimum cost of a flow $F$ in $\mathcal{I}^+_\Upsilon$ which satisfies the conditions stipulated above; this concludes the description of the algorithm.

The running time of the algorithm can be upper-bounded by the number of vertices in the input digraph times the cost of processing each vertex, whereas the latter is dominated by $\cmax^{(dw)^{\bigoh(d^4w^4)}}\cdot 2^{d^2w^2\cdot (\alpha+\cmax\cdot d^2w^2)} \cdot \alpha^{w^2d^2}$, i.e., by $2^{(d\cdot w\cdot \alpha\cdot \cmax)^{\bigoh(d^4w^4)}}$. The algorithm can be made constructive in the same way as in Theorem~\ref{thm:stcwfpt}. For correctness, we argue that if the input instance admits a flow assignment $Q$ of at most $\alpha$ agents with some minimum cost, say $p$, then the algorithm will compute a flow assignment with the same cost $p$. Towards this goal, we observe that at each vertex $v$ considered in the leaf-to-root pass made by the dynamic program, $Q$ will correspond to a unique snapshot $\Upsilon$ of $v$. At each leaf $v$ of $T$, $\Rec(v)(\Upsilon)$ must be equal to the cost incurred by $Q$ on the arcs in $\partial_v$ due to the nature of our brute-force computation of the records for trees and the optimality of $Q$. Moreover, for each non-leaf node $v$ it holds that as long as we have correctly computed the records for each of its children, the traversal of $Q$ via the arcs in $\partial_v$ and $\partial_w$ for each child of $w$ identifies a unique valid flow assignment $F$ in $\mathcal{I}^+_\Upsilon$. Moreover, from $Q$ we can also recover a unique set $\Imp$ of important agents for $v$ as well as a unique mapping $\beta$ which specifies how many agents were not routed among those with at least one endpoint in each of the children of $v$. These sets altogether define a snapshot $\Psi^w_F$ for each child $w$ of $v$. In that case, however, $Q$ must indeed incur a cost of $b(v)+\sum_{w\text{ is a child of }v}\Rec(w)(\Psi^w_F)$ over all arcs in $\mathcal{I}_\Upsilon$, as desired.
\end{proof}

\section{Concluding Remarks}
Our results provide an essentially comprehensive complexity landscape for the problem of computing system-optimal flow assignments in atomic congestion games, closing a gap in the literature that contrasts with the significant attention other aspects of congestion games have received to date. We remark that our tractability results only require the input network to have the necessary structural properties and do not impose any restrictions on the possible origins and destinations of the agents. Moreover, all of the obtained algorithms can also be used to compute Nash-equilibria in atomic congestion games as long as an upper-bound on the cost of the flow is provided in the input. Future work could also consider the recently proposed setting of having some agents follow a greedily computed route~\cite{SharonARBS18}.

Another interesting avenue for future work would be to resolve the complexity of the min-max variant of the problem (i.e., \mproblem) on well-structured networks of unbounded degree. This problem is left open even on stars when parameterized by $\cmax+\alpha$, and we believe novel ideas will be required to breach this barrier; in particular, the techniques for solving the maximization variant of the related \textsc{Arc Disjoint Paths} problem on stars~\cite{ErlebachJansen01} do not generalize to \mproblem. As a longer-term goal, one would be interested in settling whether Theorem~\ref{thm:msoacfpt} could be lifted towards an analog of Theorem~\ref{thm:stcwfpt} that relies on the same structural measures of the network.

\section{Acknowledgements}

The first, second, and fourth authors were supported by the Austrian Science Fund (FWF, project Y1329). The third author was supported by SERB-DST via grants MTR/2020/000497 and CRG/2022/009400.

\bibliographystyle{plain}
\bibliography{aaai24bib}

\begin{thebibliography}{10}

\bibitem{AbuMonsharA22}
Anees Abu{-}Monshar and Ammar Al{-}Bazi.
\newblock A multi-objective centralised agent-based optimisation approach for
  vehicle routing problem with unique vehicles.
\newblock {\em Appl. Soft Comput.}, 125:109187, 2022.

\bibitem{AckermannRV06}
Heiner Ackermann, Heiko R{\"{o}}glin, and Berthold V{\"{o}}cking.
\newblock On the impact of combinatorial structure on congestion games.
\newblock In {\em 47th Annual {IEEE} Symposium on Foundations of Computer
  Science {(FOCS} 2006)}, pages 613--622. {IEEE} Computer Society, 2006.

\bibitem{AnshelevichDKTWR04}
E.~Anshelevich, A.~Dasgupta, J.~Kleinberg, E.~Tardos, T.~Wexler, and
  T.~Roughgarden.
\newblock The price of stability for network design with fair cost allocation.
\newblock In {\em 45th Annual IEEE Symposium on Foundations of Computer
  Science}, pages 295--304, 2004.

\bibitem{AshlagiMT07}
Itai Ashlagi, Dov Monderer, and Moshe Tennenholtz.
\newblock Learning equilibrium in resource selection games.
\newblock In {\em Proc. of the 22nd {AAAI} Conference on Artificial
  Intelligence}, pages 18--23, 2007.

\bibitem{AwerbuchAE05}
Baruch Awerbuch, Yossi Azar, and Amir Epstein.
\newblock The price of routing unsplittable flow.
\newblock In {\em Proc. of the 37th Annual {ACM} Symposium on Theory of
  Computing}, pages 57--66, 2005.

\bibitem{BerwangerDHKO12}
Dietmar Berwanger, Anuj Dawar, Paul Hunter, Stephan Kreutzer, and Jan
  Obdrz{\'{a}}lek.
\newblock The dag-width of directed graphs.
\newblock {\em J. Comb. Theory, Ser. {B}}, 102(4):900--923, 2012.

\bibitem{BoehmerE20}
Niclas Boehmer and Edith Elkind.
\newblock Individual-based stability in hedonic diversity games.
\newblock In {\em Proc. of the 34th {AAAI} Conference on Artificial
  Intelligence}, pages 1822--1829, 2020.

\bibitem{BoehmerE20b}
Niclas Boehmer and Edith Elkind.
\newblock Stable roommate problem with diversity preferences.
\newblock In {\em Proc. of the 29th International Joint Conference on
  Artificial Intelligence}, pages 96--102, 2020.

\bibitem{DBLP:conf/aaai/BrandGKM24}
Cornelius Brand, Robert Ganian, Subrahmanyam Kalyanasundaram, and Fionn
  {Mc~Inerney}.
\newblock The complexity of optimizing atomic congestion.
\newblock In {\em Proc. of the 38th {AAAI} Conference on Artificial
  Intelligence}, pages 20044--20052, 2024.

\bibitem{ChanCKKP22}
Timothy F.~N. Chan, Jacob~W. Cooper, Martin Kouteck{\'{y}}, Daniel Kr{\'{a}}l,
  and Krist{\'{y}}na Pek{\'{a}}rkov{\'{a}}.
\newblock Matrices of optimal tree-depth and a row-invariant parameterized
  algorithm for integer programming.
\newblock {\em {SIAM} J. Comput.}, 51(3):664--700, 2022.

\bibitem{chen2020path}
Zhibin Chen, Xi~Lin, Yafeng Yin, and Meng Li.
\newblock Path controlling of automated vehicles for system optimum on
  transportation networks with heterogeneous traffic stream.
\newblock {\em Transportation Research Part C: Emerging Technologies},
  110:312--329, 2020.

\bibitem{ChristodoulouK05}
George Christodoulou and Elias Koutsoupias.
\newblock The price of anarchy of finite congestion games.
\newblock In {\em Proc. of the 37th Annual {ACM} Symposium on Theory of
  Computing}, pages 67--73, 2005.

\bibitem{CominettiS0M19}
Roberto Cominetti, Marco Scarsini, Marc Schr{\"{o}}der, and Nicol{\'{a}}s
  E.~Stier Moses.
\newblock Price of anarchy in stochastic atomic congestion games with affine
  costs.
\newblock In {\em Proc. of the 2019 {ACM} Conference on Economics and
  Computation}, pages 579--580, 2019.

\bibitem{CyganFKLMPPS15}
Marek Cygan, Fedor~V. Fomin, Lukasz Kowalik, Daniel Lokshtanov, D{\'{a}}niel
  Marx, Marcin Pilipczuk, Michal Pilipczuk, and Saket Saurabh.
\newblock {\em Parameterized Algorithms}.
\newblock Springer, 2015.

\bibitem{DahiyaFPS21}
Yogesh Dahiya, Fedor~V. Fomin, Fahad Panolan, and Kirill Simonov.
\newblock Fixed-parameter and approximation algorithms for {PCA} with outliers.
\newblock In {\em Proc. of the 38th International Conference on Machine
  Learning, {ICML} 2021}, volume 139 of {\em Proc. of Machine Learning
  Research}, pages 2341--2351, 2021.

\bibitem{Diestel12}
Reinhard Diestel.
\newblock {\em Graph Theory, 4th Edition}, volume 173 of {\em Graduate texts in
  mathematics}.
\newblock Springer, 2012.

\bibitem{DowneyF13}
Rodney~G. Downey and Michael~R. Fellows.
\newblock {\em Fundamentals of Parameterized Complexity}.
\newblock Texts in Computer Science. Springer, 2013.

\bibitem{DvorakEGKO21}
Pavel Dvor{\'{a}}k, Eduard Eiben, Robert Ganian, Dusan Knop, and Sebastian
  Ordyniak.
\newblock The complexity landscape of decompositional parameters for {ILP:}
  programs with few global variables and constraints.
\newblock {\em Artif. Intell.}, 300:103561, 2021.

\bibitem{EibenGKO19}
Eduard Eiben, Robert Ganian, Dusan Knop, and Sebastian Ordyniak.
\newblock Solving integer quadratic programming via explicit and structural
  restrictions.
\newblock In {\em Proc. of the 33rd {AAAI} Conference on Artificial
  Intelligence}, pages 1477--1484, 2019.

\bibitem{ErlebachJansen01}
Thomas Erlebach and Klaus Jansen.
\newblock The maximum edge-disjoint paths problem in bidirected trees.
\newblock {\em {SIAM} J. Discrete Math.}, 14(3):326--355, 2001.

\bibitem{FabrikantPT04}
Alex Fabrikant, Christos~H. Papadimitriou, and Kunal Talwar.
\newblock The complexity of pure nash equilibria.
\newblock In {\em Proc. of the 36th Annual {ACM} Symposium on Theory of
  Computing}, pages 604--612, 2004.

\bibitem{FleszarMS18}
Krzysztof Fleszar, Matthias Mnich, and Joachim Spoerhase.
\newblock New algorithms for maximum disjoint paths based on tree-likeness.
\newblock {\em Math. Program.}, 171(1-2):433--461, 2018.

\bibitem{FuchsleMNR22}
Eugen F{\"{u}}chsle, Hendrik Molter, Rolf Niedermeier, and Malte Renken.
\newblock Delay-robust routes in temporal graphs.
\newblock In {\em 39th International Symposium on Theoretical Aspects of
  Computer Science, {STACS} 2022}, volume 219 of {\em LIPIcs}, pages
  30:1--30:15. Schloss Dagstuhl - Leibniz-Zentrum f{\"{u}}r Informatik, 2022.

\bibitem{GanianHKSS22}
Robert Ganian, Thekla Hamm, Dusan Knop, Simon Schierreich, and Ondrej
  Such{\'{y}}.
\newblock Hedonic diversity games: {A} complexity picture with more than two
  colors.
\newblock In {\em Proc. of the 36th {AAAI} Conference on Artificial
  Intelligence}, pages 5034--5042, 2022.

\bibitem{GanianHKOS22}
Robert Ganian, Thekla Hamm, Viktoriia Korchemna, Karolina Okrasa, and Kirill
  Simonov.
\newblock The complexity of k-means clustering when little is known.
\newblock In {\em International Conference on Machine Learning, {ICML} 2022},
  volume 162 of {\em Proc. of Machine Learning Research}, pages 6960--6987,
  2022.

\bibitem{GanianKOS18}
Robert Ganian, Iyad~A. Kanj, Sebastian Ordyniak, and Stefan Szeider.
\newblock Parameterized algorithms for the matrix completion problem.
\newblock In {\em Proc. of the 35th International Conference on Machine
  Learning, {ICML} 2018}, volume~80 of {\em Proc. of Machine Learning
  Research}, pages 1642--1651, 2018.

\bibitem{GanianKS22}
Robert Ganian, Eun~Jung Kim, and Stefan Szeider.
\newblock Algorithmic applications of tree-cut width.
\newblock {\em {SIAM} J. Discret. Math.}, 36(4):2635--2666, 2022.

\bibitem{GanianK21}
Robert Ganian and Viktoriia Korchemna.
\newblock The complexity of bayesian network learning: Revisiting the
  superstructure.
\newblock In {\em Advances in Neural Information Processing Systems 34: Annual
  Conference on Neural Information Processing Systems 2021, NeurIPS 2021},
  pages 430--442, 2021.

\bibitem{GanianK22}
Robert Ganian and Viktoriia Korchemna.
\newblock Slim tree-cut width.
\newblock {\em Algorithmica}, 86(8):2714--2738, 2024.

\bibitem{GanianO18}
Robert Ganian and Sebastian Ordyniak.
\newblock The complexity landscape of decompositional parameters for {ILP}.
\newblock {\em Artif. Intell.}, 257:61--71, 2018.

\bibitem{GanianO21}
Robert Ganian and Sebastian Ordyniak.
\newblock The power of cut-based parameters for computing edge-disjoint paths.
\newblock {\em Algorithmica}, 83(2):726--752, 2021.

\bibitem{GanianOS17}
Robert Ganian, Sebastian Ordyniak, and Ramanujan Sridharan.
\newblock On structural parameterizations of the edge disjoint paths problem.
\newblock In {\em 28th International Symposium on Algorithms and Computation,
  {ISAAC} 2017}, volume~92 of {\em LIPIcs}, pages 36:1--36:13. Schloss Dagstuhl
  - Leibniz-Zentrum f{\"{u}}r Informatik, 2017.

\bibitem{GozupekOPSS17}
Didem G{\"{o}}z{\"{u}}pek, Sibel {\"{O}}zkan, Christophe Paul, Ignasi Sau, and
  Mordechai Shalom.
\newblock Parameterized complexity of the {MINCCA} problem on graphs of bounded
  decomposability.
\newblock {\em Theor. Comput. Sci.}, 690:91--103, 2017.

\bibitem{GruttemeierK22}
Niels Gr{\"{u}}ttemeier and Christian Komusiewicz.
\newblock Learning bayesian networks under sparsity constraints: {A}
  parameterized complexity analysis.
\newblock {\em J. Artif. Intell. Res.}, 74:1225--1267, 2022.

\bibitem{GruttemeierKM21}
Niels Gr{\"{u}}ttemeier, Christian Komusiewicz, and Nils Morawietz.
\newblock On the parameterized complexity of polytree learning.
\newblock In {\em Proc. of the 30th International Joint Conference on
  Artificial Intelligence}, pages 4221--4227, 2021.

\bibitem{HarksHKMS22}
Tobias Harks, Mona Henle, Max Klimm, Jannik Matuschke, and Anja Schedel.
\newblock Multi-leader congestion games with an adversary.
\newblock In {\em Proc. of the 36th {AAAI} Conference on Artificial
  Intelligence}, pages 5068--5075, 2022.

\bibitem{Huang0KD22}
Taoan Huang, Jiaoyang Li, Sven Koenig, and Bistra Dilkina.
\newblock Anytime multi-agent path finding via machine learning-guided large
  neighborhood search.
\newblock In {\em Proc. of the 36th {AAAI} Conference on Artificial
  Intelligence}, pages 9368--9376, 2022.

\bibitem{hunter2008digraph}
Paul Hunter and Stephan Kreutzer.
\newblock Digraph measures: Kelly decompositions, games, and orderings.
\newblock {\em Theoretical Computer Science}, 399(3):206--219, 2008.

\bibitem{jalota2023balancing}
Devansh Jalota, Kiril Solovey, Matthew Tsao, Stephen Zoepf, and Marco Pavone.
\newblock Balancing fairness and efficiency in traffic routing via interpolated
  traffic assignment.
\newblock {\em Autonomous Agents and Multi-Agent Systems}, 37(2):32, 2023.

\bibitem{JohnsonRST01}
Thor Johnson, Neil Robertson, Paul~D. Seymour, and Robin Thomas.
\newblock Directed tree-width.
\newblock {\em J. Comb. Theory, Ser. {B}}, 82(1):138--154, 2001.

\bibitem{KoanaFN23}
Tomohiro Koana, Vincent Froese, and Rolf Niedermeier.
\newblock The complexity of binary matrix completion under diameter
  constraints.
\newblock {\em J. Comput. Syst. Sci.}, 132:45--67, 2023.

\bibitem{KoanaKNNZ21}
Tomohiro Koana, Viatcheslav Korenwein, Andr{\'{e}} Nichterlein, Rolf
  Niedermeier, and Philipp Zschoche.
\newblock Data reduction for maximum matching on real-world graphs: Theory and
  experiments.
\newblock {\em {ACM} J. Exp. Algorithmics}, 26:1.3:1--1.3:30, 2021.

\bibitem{KorhonenP15}
Janne~H. Korhonen and Pekka Parviainen.
\newblock Tractable bayesian network structure learning with bounded vertex
  cover number.
\newblock In {\em Advances in Neural Information Processing Systems 28: Annual
  Conference on Neural Information Processing Systems 2015}, pages 622--630,
  2015.

\bibitem{KunniyurS03}
S.~Kunniyur and R.~Srikant.
\newblock End-to-end congestion control schemes: utility functions, random
  losses and ecn marks.
\newblock {\em IEEE/ACM Trans. on Networking}, 11(5):689--702, 2003.

\bibitem{MarchesiC019}
Alberto Marchesi, Matteo Castiglioni, and Nicola Gatti.
\newblock Leadership in congestion games: Multiple user classes and
  non-singleton actions.
\newblock In {\em Proc. of the 28th International Joint Conference on
  Artificial Intelligence}, pages 485--491, 2019.

\bibitem{MarxW14}
D{\'{a}}niel Marx and Paul Wollan.
\newblock Immersions in highly edge connected graphs.
\newblock {\em {SIAM} J. Discret. Math.}, 28(1):503--520, 2014.

\bibitem{MeirTBK12}
Reshef Meir, Moshe Tennenholtz, Yoram Bachrach, and Peter~B. Key.
\newblock Congestion games with agent failures.
\newblock In {\em Proc. of the 26th {AAAI} Conference on Artificial
  Intelligence}, 2012.

\bibitem{sparsity}
Jaroslav Nesetril and Patrice~Ossona de~Mendez.
\newblock {\em Sparsity - Graphs, Structures, and Algorithms}, volume~28 of
  {\em Algorithms and combinatorics}.
\newblock Springer, 2012.

\bibitem{Niedermeier06}
Rolf Niedermeier.
\newblock {\em Invitation to Fixed-Parameter Algorithms}.
\newblock Oxford Lecture Series in Mathematics and its Applications. Oxford
  University Publishing, Oxford, 2006.

\bibitem{OrdyniakS13}
Sebastian Ordyniak and Stefan Szeider.
\newblock Parameterized complexity results for exact bayesian network structure
  learning.
\newblock {\em J. Artif. Intell. Res.}, 46:263--302, 2013.

\bibitem{PhamHVN22}
Quang~Anh Pham, Minh~Ho{\`{a}}ng H{\`{a}}, Duy~Manh Vu, and Huy~Hoang Nguyen.
\newblock A hybrid genetic algorithm for the vehicle routing problem with
  roaming delivery locations.
\newblock In {\em Proc. of the 32nd International Conference on Automated
  Planning and Scheduling}, pages 297--306, 2022.

\bibitem{PorschenSSW14}
Stefan Porschen, Tatjana Schmidt, Ewald Speckenmeyer, and Andreas Wotzlaw.
\newblock {XSAT} and {NAE-SAT} of linear {CNF} classes.
\newblock {\em Discret. Appl. Math.}, 167:1--14, 2014.

\bibitem{RobertsonS86}
Neil Robertson and Paul~D. Seymour.
\newblock Graph minors. {II.} algorithmic aspects of tree-width.
\newblock {\em J. Algorithms}, 7(3):309--322, 1986.

\bibitem{Rosenthal73}
Robert~W. Rosenthal.
\newblock A class of games possessing pure-strategy nash equilibria.
\newblock {\em International J. of Game Theory}, 1973.

\bibitem{Schmidt2010d}
Tatjana Schmidt.
\newblock {\em Computational complexity of SAT, {XSAT} and {NAE-SAT} for linear
  and mixed Horn {CNF} formulas}.
\newblock PhD thesis, University of Cologne, 2010.

\bibitem{SeymourT93}
P.~D. Seymour and R.~Thomas.
\newblock Graph searching and a min-max theorem for tree-width.
\newblock {\em Journal of Combinatorial Theory, Series B}, 58(1):22--33, 1993.

\bibitem{sharon2021alleviating}
Guni Sharon.
\newblock Alleviating road traffic congestion with artificial intelligence.
\newblock In {\em IJCAI}, pages 4965--4969, 2021.

\bibitem{SharonARBS18}
Guni Sharon, Michael Albert, Tarun Rambha, Stephen~D. Boyles, and Peter Stone.
\newblock Traffic optimization for a mixture of self-interested and compliant
  agents.
\newblock In {\em Proc. of the 32nd {AAAI} Conference on Artificial
  Intelligence}, pages 1202--1209, 2018.

\bibitem{sharon2017real}
Guni Sharon, Josiah~P Hanna, Tarun Rambha, Michael~W Levin, Michael Albert,
  Stephen~D Boyles, and Peter Stone.
\newblock Real-time adaptive tolling scheme for optimized social welfare in
  traffic networks.
\newblock In {\em Proc. of the 16th International Conference on Autonomous
  Agents and Multiagent Systems}, 2017.

\bibitem{sharon2017network}
Guni Sharon, Michael~W Levin, Josiah~P Hanna, Tarun Rambha, Stephen~D Boyles,
  and Peter Stone.
\newblock Network-wide adaptive tolling for connected and automated vehicles.
\newblock {\em Transportation Research Part C: Emerging Technologies},
  84:142--157, 2017.

\bibitem{0040XPRT22}
Kai Wang, Lily Xu, Andrew Perrault, Michael~K. Reiter, and Milind Tambe.
\newblock Coordinating followers to reach better equilibria: End-to-end
  gradient descent for stackelberg games.
\newblock In {\em Proc. of the 36th {AAAI} Conference on Artificial
  Intelligence}, pages 5219--5227, 2022.

\end{thebibliography}

\end{document}